\documentclass[13pt]{article}

\usepackage{graphicx}

\usepackage{amsthm,amsmath,natbib,amsfonts}
\usepackage{xcolor}
\usepackage{subcaption}
\usepackage{color}
\usepackage{makecell}
\usepackage{algpseudocode}
\usepackage{algorithm}
\usepackage{geometry}
\usepackage{booktabs}
\usepackage{siunitx}
\usepackage{threeparttable}
\newtheorem{theorem}{{\bf \sc Theorem}}
\newtheorem{proposition}{{\bf \sc Proposition}}
\newtheorem{lemma}[proposition]{{\bf \sc Lemma}}

\newtheorem{definition}{{\bf \sc Definition}}
\newtheorem{remark}{{\bf \sc Remark}}
\newtheorem{result}{{\bf \sc Result}}

\newtheorem{assumption}{{\bf \sc Assumption}}
\theoremstyle{remark} \newtheorem{example}{{\bf \sc Example}}

\newcommand{\bge}{\begin{equation}}
\newcommand{\ene}{\end{equation}}

\newcommand{\R}{ \mathbb{R}}
\newcommand{\W}{\mathcal{W}}

\newcommand{\ubar}[1]{\text{\b{$#1$}}}

\begin{document}

%\begin{frontmatter}

% "Title of the paper"
\title{Decentralized Pure Exchange Processes on Networks\footnote{Previous versions of this paper have circulated under the titles ``A Second Welfare Theorem for Pure Exchange Processes on Networks'' and ``Edgeworth trading on Networks''.
We thank Malbor Asllani, Timoteo Carletti, Simone Cerreia Vioglio, Luca Chiantini, Satoshi Fukuda, Edoardo Gallo, Benjamin Golub, Michael  K{\"o}nig, Thomas Mariotti, Alessandro Pavan, Joel Robbins, Alex Teytelboym and Alexander Wolitzky for helpful comments.
P.\ Pin acknowledges funding from the Italian Ministry of Education Progetti di Rilevante Interesse Nazionale (PRIN) grant  2017ELHNNJ.}}
\author{%%%% First author details
{\sc Daniele Cassese} \\[2pt]
Emmanuel College, University of Cambridge, UK\\
{dc554@cam.ac.uk}\\[2pt]
\\[6pt]
%%%%%%%
{\sc and}\\[6pt]
%%%%%%% Third author details
{\sc Paolo Pin}\\[2pt]
Department of Economics and Statistics, Universit\`a di Siena, Italy \\
BIDSA, Universit\`a Bocconi, Italy\\
{paolo.pin@unisi.it}}
\date{}

\maketitle

\begin{abstract}
We define a class of pure exchange Edgeworth trading processes that under minimal assumptions converge to a stable set in the space of allocations, and  characterise the Pareto set of these processes. Choosing a specific process belonging to this class, that we define \emph{fair trading}, we analyse the trade dynamics  between agents located on a weighted network. We determine the conditions under which there always exists a one-to-one map between the set of networks and the set of limit points of the dynamics. This result is used to understand what is the effect of the network topology on the trade dynamics and on the final allocation. % for the case of Cobb-Douglas utility function. 
We find that the positions in the network affect the distribution of the utility gains, given the initial allocations.
\end{abstract}


%\begin{keyword}
%\kwd{}
%\kwd{}
%\end{keyword}

%\end{frontmatter}

\section{Introduction}

This paper contributes to the literature on the dynamics of trade, providing a model of trade where agents are located on a network. There are many reasons to consider the network structure of trading opportunities, starting with the fact that real trades are shaped and influenced by the structure of relationships between agents: not everybody interacts with everybody else due to geography, social ties, technological compatibility. 
Trade on networks is a very active area of research in economics, for an exhaustive review see \citep{Manea_for}. The main difference between the contributions reviewed in \cite{Manea_for} and our work is that we do not explicitly model strategic interactions among agents, instead we focus on the characterisation of the dynamics of trade on a fixed network with a tractable convergent dynamical systems.
Inside this framework we prove a version of the Second Welfare Theorem for networks, contributing to the analysis of the effect of the network structure on final allocations and the distribution of welfare.

%\subsection{Relation with the literature}
\noindent
In the Walrasian competitive equilibrium all trading in decentralized exchange takes place at the final equilibirum prices, while in actual market transactions agents discover equilibrium prices only by making mutually advantageous transaction at disequilibrium prices \citep{Foley10}. In order to circumvent the impasse given by the impossibility of a real price dynamics, the fictitious figure of the ``auctioneer" has to be introduced: in \emph{t\^{a}tonnement} models agents constantly recontract instead of trading and so only prices change out of equilibrium while quantities are fixed \citep{Fisher03}. 

\noindent
This paper adopts a different perspective, placing itself in the literature of out-of-equilibrium dynamics that started in the early sixties (see \citealp{petri_hahn} for a review). These models are called \emph{non-t\^{a}tonnement} processes, or trading processes. \cite{Uzawa62} and \cite{Hahn62} introduced the so called ``Edegworth process", where both prices and quantities adjust along the process: equilibrium is path-dependent, and out of equilibrium dynamics change the equilibrium set, while in the Walrasian case equilibrium is determined solely by the initial holdings and is independent of the path. Edgeworth processes rest on one fundamental assumption: trade takes place if and only if there is an increase in utility by trading. Both \cite{Uzawa62} and \cite{Hahn62} show that under standard assumption on the structure of preferences and on the space of goods in the economy, these processes approach in the limit a Pareto optimum.

\noindent
We model out of equilibrium dynamics, with quantity and prices adjustment using a version of the Edgeworth barter process. We define a class of trading processes that under a limited number of assumptions converge to equilibrium. Prices are the instantaneous rate of exchange between goods, and they can change at any moment along the process. Moreover, there is no gravitation towards equilibrium prices because equilibria are path-dependent. 

\noindent
{Our work is related to the literature on planning procedures  \citep{Dreze1971, Malinvaud1972} and in particular to the work of \cite{Cornet1983} on the neutrality of planning procedures, which we discuss in more details in section 5.}

%In the family of the trading processes that satisfy our assumption we restrict attention to a very specific form of trade, and we adopt the egalitarian solution as the rule according to which agents share utility gains when they trade. 
\noindent
The main novelty with respect to the literature on out-of-equilibrium dynamics is the fact that only connected agents can trade: we introduce a static, weighted network determining who can trade with whom. 

\noindent
Among the works on dynamical networks, our paper relates with \cite{CowanJonard04} and \cite{koniga2017endogenous} who model knowledge diffusion as a barter process: agents meet their neighbours repeatedly and in case they have a differential in two dimensions of knowledge they trade, each receiving a constant share of the knowledge differential. %As utility increases with knowledge levels at the end of the trade both parties have increased their utility.  %They show that diffusion of knowledge is maximized where the proportion of links between an agent and other agents not in her neighbourhood is between 1 and 10 percent of all direct links between agent pairs.

\noindent 
{\color{black} Other related works include \cite{Flam2019} who study the emergence of price taking behaviour modelling trade as a sequence of bilateral exchanges where agents only trade if each exchange increases both agents' utility. Contrary to our case the network structure of agents matching is not explored but the author shows that equilibria can be path-dependent and are affected by the matching order.}

\noindent
Our model does not consider strategic interactions: we do not have any market game and agents do not trade with all others simultaneously but only engage in bilateral exchanges with their network contacts, which differentiates our approach from \cite{Ghosal2004}. Moreover, we do not allow for trade frictions as for example \cite{Fleiner2019}.

%, who model retrading between a finite set of agents engaging in a pure exchange process, showing that if agents play Nash at each round of retrading (such that each trade is Pareto-improving) they are able to converge to an allocation in the Pareto set.

%\noindent
%In our model we do not allow for agents to be out of the market: every agent will trade as the network is connected, hence our process is non-neutral, as proven by \cite{bottazzi1994accessibility}.

\noindent
{\color{black} Our model in principle can be applied to large numbers of players, even if in the current work we provide examples of small networks only. \cite{Axtell05} proves that decentralized exchange processes of the same class of our model have polynomial computational complexity, performing much better than Walrasian models which can be exponential in the worst case. }

The paper is structured as follows: in section 2 we define our family of bilateral trading processes, and we provide a characterisation of the Pareto set to which these processes converge. In section 3 we expose the trading rule of choice, namely the egalitarian rule, proving that the trade so defined belongs to the family of trades of our interest. In section 4 we extend trading to more agents, and we introduce the network structure as a weighted network. In section 5 we prove our version of the Second Welfare Theorem for networks, while in section 6 we provide a numerical example with Cobb-Douglas preferences, illustrating our result and investigating the impact of networks on equilibria. Finally section 7 discusses the Pairwise Stability of some network configurations under the fair trade dynamics.

\section{The model} 

\subsection{Pure exchange}\label{para: pure_exch}

There are $n \geq 2$ agents, we will generally refer to an agent $i \in \{ 1, \dots , n \} \equiv N$,
and $m \geq 2$ goods, and to a good $k \in \{ 1, \dots , m \} \equiv M$. Agents can only have non-negative quantities of each good, and we are considering a pure exchange economy with no production,  so that total resources in the economy are fixed and given by the sum of the agents' endowments. The endowment of agent $i$ is a point in the positive orthant of $\R^m$, call this space $\R_{+}^m $, where the $k-$th coordinate represents the quantity of good $k$. 
Assume time $t$ is continuous, with $t \in (0, \infty)$ and goods are infinitely divisible, and let $x_{ik,t}$ be the endowment of agent $i$ at time $t$ for good $k$.
In this way $ {x}_{i,t} \in \R^m$ is the $m$--dimensional vector of agent $i$'s endowment at time $t$, while $ {x}_{k,t} \in \R^n$ is the $n$--dimensional vector of all agents' endowments of good $k$ at time $t$.
As we assumed there is no production, nor can the goods be disposed of, the sum of the elements of each such vector $ {x}_{k,t} $ is constant in time.
The initial allocation of the economy is then represented by the $n$ vectors of agents' endowment at time zero, call it  $x_0 =  \{ {x}_{1,0} , \dots , {x}_{n,0} \}$. 
All agents' allocations at a given point in time can then be represented by an $(m \times n)$ matrix with all non-negative entries, call it ${\bf X}_t$. In the following we may not express the time variable $t$, when it does not create ambiguity. An unrestricted state of the economy at any time $t$ is a point in the positive orthant of an $\R^{m \times n}$ space, given by the Cartesian product $(\R^m)^n$ . As we assumed that resources are fixed in the economy at a point $w \in \R^m$ (where the $k$-th coordinate is the total quantity of good $k$ in the economy), the state space of our interest is a subset of $\R_+^{m \times n}$, call it $E=\{ x \in \R_+^{m \times n} : \sum x_i = w_i \}$, which is an open subset of an affine subspace with compact closure in $\R^{m \times n}$ \citep{smale_I}.

\begin{assumption}
Any agent $i$ is characterised by a twice continuously differentiable, strictly increasing utility function $U_i$ from $\R_+^m$ to $\R$.
\end{assumption}

%\begin{assumption}
%Preferences are assumed to be strictly convex.
%\end{assumption}

\noindent
Given $ x_t \in E $, a point in the space of the economy at some point in time $t$, call ${U} (x_t)$ its corresponding $n$--dimensional vector of utilities. 
Define $\mu_{ik,t} \equiv \partial U_i ( {x}_{i,t} ) / \partial x_{ik} $ the marginal utility of agent $i$, with endowment ${x}_{i,t}$, with respect to good $k$, and ${\mu}_{i,t}$ the gradient of the utility function for agent $i$ at time $t$, that is the vector of all her marginal utilities. All individual gradients are represented by an $m \times n$ matrix of all the marginal utilites at a given point in time, call it  ${\bf M}_t$. 
The vector of strictly positive marginal utilities ${\mu}_{i,t}$, is proportional to any vector of marginal rates of substitutions with respect to any good $\ell \in \{ 1, \dots , m \}$.
{It is important to stress that for the rest of the paper we use a cardinal notion of utility, because this is the structure on which we build on our out--of--equilibrium dynamical process.}

\begin{assumption}
Agents' preferences are represented by cardinal utility functions.
\end{assumption}

\noindent
In the pure exchange economy defined above, the contract curve is given by the set of all those allocation where all marginal utilities are proportional.

\begin{definition}
The Pareto Set $\mathcal{W}$ of the pure exchange economy is defined as:
\begin{equation}\label{contract_c}
\W = \left\{ {\bf X} : \forall \ i,j \in N, \ \exists k \in \R , \ k \ne 0,  \ s.t. \ \mu_{i} ({x}_i) = k \mu_{j} ({x}_j) \right\} \ \ .
\end{equation}
\end{definition}

\begin{proposition}[\citealt{smale_I}] \label{smale_hom}  If the utility function is monotonic and indifference curves are convex, then the set of Pareto Optima is homoeomorphic to a closed $(n-1)$ simplex.
\end{proposition}

\noindent
For a proof of Proposition \ref{smale_hom} see \cite{smale_I}. %, or the appendix (\ref{A4}). 
%Remember that a diffeomorphism implies an homeomorphism but not viceversa.
The assumptions in  Proposition \ref{smale_hom} are standard in economics. Furthermore if we assume convexity of the utility function and of the commodity space we have that the set of Pareto Optima is diffeomorphic to a closed $(n-1)$ simplex.
It has been shown that if preferences are $C^2$ and convex it is possible to find utility representations that admit a convex space, for an exhaustive discussion and proofs see \cite{diff_approach}.
Note that in our case the assumptions of Proposition  \ref{smale_hom} are satisfied: the state space of interest is an open subset of an affine subspace with compact closure in $\R^{m \times n}$ \citep{smale_I}. The convexity assumption makes the problem much easier to deal with, but in case this assumption is relaxed we can still characterise the Pareto Set, that will be an ($n-1$) stratified set, that is a manifold with borders and corners, see \cite{YHWan} and \cite{de_melo}.
Note finally that adding an error term to equation \eqref{contract_c} we get a diffusion process similar to the one analyzed by \cite{anderson2004noisy},  generalized to networks by \cite{bervoets2016learning} and, outside economics, by \cite{robert2016dynamics}.

\subsection{Trading}\label{para: trading}

Define \emph{trading} between agents in $N$ as a continuous dynamic over the endowments, which is based on marginal utilities.
Formally it will be a set of differential equations of the form:
\begin{equation}    \label{def_dyn}
\frac{d {x}_{i,t}}{d t} = f_i \left( {\bf M}_t \right) \ \ ,
\end{equation}
where function $f_i$ from $\R_+^{n \times m}$ to $\R^m$, satisfies the following 3 assumptions, for any set $ {\bf M}_t = \left( {\mu}_{1,t} , {\mu}_{2,t} , \dots {\mu}_{n,t} \right)$ of feasible marginal utilities:

\begin{itemize}
\item {\bf Zero sum:} the sum $ \sum_{i=1}^n {f}_i$ is equal to the null vector ${0}$.
%\end{assumption} 
%\begin{assumption}
\item {\bf Trade:} if there are at least two vectors of marginal utilities, ${\mu}_{i,t}$ and ${\mu}_{j,t}$, which are linearly independent, then at least one between ${f}_i $ and ${f}_j$ is different from ${0}$.
%\end{assumption}
%\begin{assumption}
\item {\bf Positive gradient:} for any agent $i$ it will always be the case that ${\mu}_{i,t} \cdot {f}_i \geq 0 $, with strictly positive sign if there is trade.
\end{itemize}

\noindent
The assumption of \emph{zero sum trade} guarantees that we are in a pure exchange economy without consumption nor production of new goods, as the amount of all the goods remain unchanged at any step of the process.
The assumption of \emph{trade} guarantees that there is actually exchange, unless we are in a Pareto optimal allocation, where the marginal rate of substitution between any two goods would be the same for any couple of agents.
Finally, the assumption of \emph{positive gradient} guarantees that any marginal exchange represents a Pareto improvement.
That is because

\begin{equation} \label{marg_ut}
\begin{split} 
\frac{d U_i}{d t} & =  \sum_{k=1}^m \frac{d U_i}{d x_{ik}} \frac{d x_{ik}}{d t}  \\
& = {\mu}_{i,t} \cdot f_i \left( {\bf M}_t \right) \geq 0 \\
\end{split}
\end{equation}

{\color{black}
\noindent 
An allocation $\bold{X}^*$ (the  $(n \times m)$ matrix representing quantities of each of the $m$ goods for each of the $n$ agents) is an \emph{equilibrium} of the system in equation (\ref{def_dyn}) if $f_i(\bold{X}^*) = 0$ for all $i = 1,\dots, n$ and a solution is a function $x(t,\bold{X}_0): \R \times \R_+^{n \times m} \rightarrow \R_+^{n \times m}$ where  $\bold{X}_0$ is the initial condition at time 0. 

\noindent
\begin{definition}
The solution $x(t, \bf{X}_0^*)$ is \emph{stable} if for every $\epsilon > 0$ there exist a $\delta$ such that:
\begin{equation}
| \bold{X}_0 - \bold{X}^*_0| \le \delta \, \, \implies |x(t,\bold{X}_0) - x(t,\bold{X}^*_0)| \le \epsilon \, \, \forall t \ge 0
\end{equation}
\end{definition}

\noindent
Generalizing \citep{Hahn82}, it is easy to show that all and only the fixed points of the dynamical system defined in Equation \eqref{def_dyn}, are Pareto optimal allocations.
That is because the function
\[
\bar{U} \left( {\bf X}_t \right) \equiv
\sum_{i=1}^n U_i ({x}_{i,t})
\]
can be seen as a potential.
It is bounded in its dominion of all possible allocations, it strictly increases as long as there is trade (i.e.~out of equilibrium), and it is stable when there are no two agents who could both profitably exchange goods between them.
At the limit $\bar{U}$ will converge for sure to a value, say  $\bar{U}^*$, corresponding to an allocation ${\bf X}^*$.
As preferences are strictly convex, there will be no trade in ${\bf X}^*$.

\noindent
The fixed points of the above dynamical system are reached by a sequence of utility increasing, infinitesimally small trades from an initial state, hence the set of the solutions of any such trade mechanism is an open subset of the Pareto Set $\W$ defined in Equation (\ref{contract_c}) \citep{smale_I}.

\noindent
Note that at this stage there are no assumptions restricting endowments not to become negative, that is to say we are not requiring a condition like $\frac{d x_{ik}}{d t} >0$ as $x_{ik} \rightarrow 0$.
This will depend on the initial endowment $ {\bf X}_t$ of the agents and or their utility functions.

\begin{assumption}
As any marginal exchange represents a Pareto improvement, we assume that any Pareto improvement starting from the initial conditions will lie in the region of non--negative endowments.
\end{assumption}

\noindent
Examples that satisfy these properties are the classical Walrasian \emph{t{\^a}tonnement} process, as well as  non-\emph{t{\^a}tonnement} processes, as can be find in \cite{Hahn82} and \cite{HRR75I,HRR75II}.

\section{\emph{Fair trading} between two agents}  \label{sec_2fair}

Let us start by considering $n=2$.
There is an entire family of trading mechanisms satisfying the very general assumptions of zero sum, trade and positive gradient. As we choose a trading mechanism we are implicitly making assumptions on some bargaining rule that has been fixed by the agents participating in the trade. This is a restriction to some extent, still we can choose  different trading mechanisms corresponding to different bargaining solutions that satisfy the assumptions. We define a mechanism that we call \emph{fair trading}, that is based on the  egalitarian solution by \citep{K77}:  whenever there is room for a Pareto improvement, agents trade if and only if they equally split the gains in utility from the trade.

{As in \cite{K77} we are assuming that utility is cardinal, in other words it does not only represent an ordering among alternatives, it also attaches a precise value to alternatives on the same indifference curve. Moreover, we are allowing for interpersonal comparisons and we are assuming that agents have full knowledge of each others' preferences. These strong assumptions imply first of all that our results are not invariant under monotone transformation of the utility function. In this respect we also do not explicitly consider strategic misrepresentation of preferences, in other words we are assuming that agents are always revealing their true utility function. It is worth remembering that \cite{Shapley1969} showed that there is no strongly individually rational ordinal solution to bilateral bargaining problems, and here we are considering a trade mechanism based on a bilateral bargaining solution, as in \cite{K77}.}

\noindent
Trading is bilateral, $N= \{1,2\}$, and $m \geq 2$ goods.
By the zero sum property we have that ${f}_1 = -{f}_2$.
We are restricting our attention to the case where marginal utility from trading is equally split among the two agents.
The Pareto improvement from trading is defined in Equation \eqref{marg_ut}, so we are requiring that:
\[
{\mu}_{1,t} \cdot f_1 \left( {\mu}_{1,t} , {\mu}_{2,t} \right) = {\mu}_{2,t} \cdot f_2 \left( {\mu}_{1,t} , {\mu}_{2,t} \right) \ \ .
\]
By the zero sum property this is satisfied if
\[
%\begin{equation}  \label{orth_cond}
\left( {\mu}_{1,t} + {\mu}_{2,t} \right) \cdot f_1 \left( {\mu}_{1,t} , {\mu}_{2,t} \right) = 0 \ \ ,
%\end{equation}
\]
which simply means that marginal trade has to be orthogonal to the sum of marginal utilities.

\noindent
There is a full sub--space of dimension $m-1$ that is orthogonal to the sum of the two marginal utilities.
Here we consider a single element that lies in the sub--plane generated by ${\mu}_{1,t}$ and $ {\mu}_{2,t}$.
We assume that trade for agent 1, ${f}_1$, is the orthogonal part of  ${\mu}_{1,t}$ with respect to ${\mu}_{1,t} + {\mu}_{2,t}$ (or the vector rejection of ${\mu}_{1,t}$ from  ${\mu}_{1,t} + {\mu}_{2,t}$).
In formulas it is
\begin{equation}  \label{fair_f}
f_1 \left( {\mu}_{1,t} , {\mu}_{2,t} \right) =
{\mu}_{1,t} - \frac{{\mu}_{1,t} \cdot \left( {\mu}_{1,t} + {\mu}_{2,t} \right) }{|{\mu}_{1,t} + {\mu}_{2,t}|^2}
\left( {\mu}_{1,t} + {\mu}_{2,t} \right) \ \ 
\end{equation}
where $| \cdot |$ is the Euclidean norm in $\R^m$.  {\color{black} Generalizing Equation (\ref{fair_f}) we call $f_i(\mu_{i,t}, \mu_{j,t})$ \emph{fair trading} between agent $i$ and $j$.}

\begin{proposition} 
The fair trading mechanism between two agents defined in Equation \eqref{fair_f} satisfies zero sum,  trade and positive gradient.
\end{proposition}

\begin{proof}
Fair trading specified in \eqref{fair_f} satisfies \emph{zero-sum}, as the instantaneous trade of one agent is equal to the additive inverse of the instantaneous trade of the other agent:

\[
f_2 \left( {\mu}_{1,t} , {\mu}_{2,t} \right) =
{\mu}_{2,t} - \frac{{\mu}_{2,t} \cdot \left( {\mu}_{1,t} + {\mu}_{2,t} \right) }{|{\mu}_{1,t} + {\mu}_{2,t}|^2}
\left( {\mu}_{1,t} + {\mu}_{2,t} \right) = -f_1 \left( {\mu}_{1,t} , {\mu}_{2,t} \right) 
\]
because
\[
f_1 \left( {\mu}_{1,t} , {\mu}_{2,t} \right) + f_2 \left( {\mu}_{1,t} , {\mu}_{2,t} \right)
= \left( {\mu}_{1,t} + {\mu}_{2,t} \right)  - \frac{ |{\mu}_{1,t} + {\mu}_{2,t}|^2}{|{\mu}_{1,t} + {\mu}_{2,t}|^2}
\left( {\mu}_{1,t} + {\mu}_{2,t} \right)
= {0} \ \ .
\]
 To check that the \emph{trade} condition is satisfied note that $f_1 \left( {\mu}_{1,t} , {\mu}_{2,t} \right) =0$ only if ${\mu}_{1,t} = k {\mu}_{2,t}$ for some $k \in \R$, that is when ${\mu}_{1,t}$ and ${\mu}_{2,t}$ are linearly dependent. 

\noindent
Positive gradient requires that:
\[
{\mu}_{1,t} \cdot \left(
{\mu}_{1,t} - \frac{{\mu}_{1,t} \cdot \left( {\mu}_{1,t} + {\mu}_{2,t} \right) }{|{\mu}_{1,t} + {\mu}_{2,t}|^2}
\left( {\mu}_{1,t} + {\mu}_{2,t} \right)  \right) \geq 0
\]
For the above inequality to be satisfied it suffices that $| {\mu}_{1,t} \cdot \left( {\mu}_{1,t} + {\mu}_{2,t} \right) | \leq |{\mu}_{1,t}| |{\mu}_{1,t} + {\mu}_{2,t}|$.% where $\| \cdot \|$ is the classical norm in $\R$.
The latter alway holds as it is the Cauchy--Schwarz inequality.  As long as  ${\mu}_{1,t}$ and $ {\mu}_{2,t}$ are linearly independent  $| {\mu}_{1,t} \cdot \left( {\mu}_{1,t} + {\mu}_{2,t} \right) | < |{\mu}_{1,t}| |{\mu}_{1,t} + {\mu}_{2,t}|$ and so $\mu_{1,t} f_1 \left( {\mu}_{1,t} , {\mu}_{2,t} \right) > 0$ as long as there is trade. When ${\mu}_{1,t}$ and $ {\mu}_{2,t}$ are linearly dependent  $| {\mu}_{1,t} \cdot \left( {\mu}_{1,t} + {\mu}_{2,t} \right) | = |{\mu}_{1,t}| |{\mu}_{1,t} + {\mu}_{2,t}|$, so $\mu_{1,t} f_1 \left( {\mu}_{1,t} , {\mu}_{2,t} \right) = 0$ when there is no trade.
\end{proof}

\noindent
Note that \emph{zero-sum}, \emph{trade} and \emph{positive gradient} would be satisfied for any $\alpha f_1 \left( {\mu}_{1,t} , {\mu}_{2,t} \right)$, with $\alpha >0 $, where the parameter $\alpha$ represents the \emph{speed} at which the dynamical system is moving, so there will be no loss in generality in assuming it equal to 1.

\noindent
So, the fair trading mechanism is a bilateral pure exchange mechanism satisfying the required three assumptions. The two agents trade over $m \geq 2$ goods,
starting from some initial allocation ${\bf X}_0 \in \R^{m \times 2}$ and evolving according to the following system of differential equations in matrix form, based on Equations \eqref{def_dyn} and \eqref{fair_f}:
\begin{equation} \label{fair_mechanism}
\frac{d {\bf X}_t}{d t} = 
\left(
{\mu}_{1,t} - \frac{{\mu}_{1,t} \cdot \left( {\mu}_{1,t} + {\mu}_{2,t} \right) }{|{\mu}_{1,t} + {\mu}_{2,t}|^2}
\left( {\mu}_{1,t} + {\mu}_{2,t} \right) ,
{\mu}_{2,t} - \frac{{\mu}_{2,t} \cdot \left( {\mu}_{1,t} + {\mu}_{2,t} \right) }{|{\mu}_{1,t} + {\mu}_{2,t}|^2}
\left( {\mu}_{1,t} + {\mu}_{2,t} \right)
  \right) \ \ .
\end{equation}
This dynamical system is well defined, as ${\mu}_{1,t}$ and ${\mu}_{2,t}$ are defined in ${\bf X}_t$, and are based on the utilities $U_1$ and $U_2$.
However, this system is not linear in ${\bf M}_t$.

\begin{figure}[h]
\begin{center}
\includegraphics*[width=6cm]{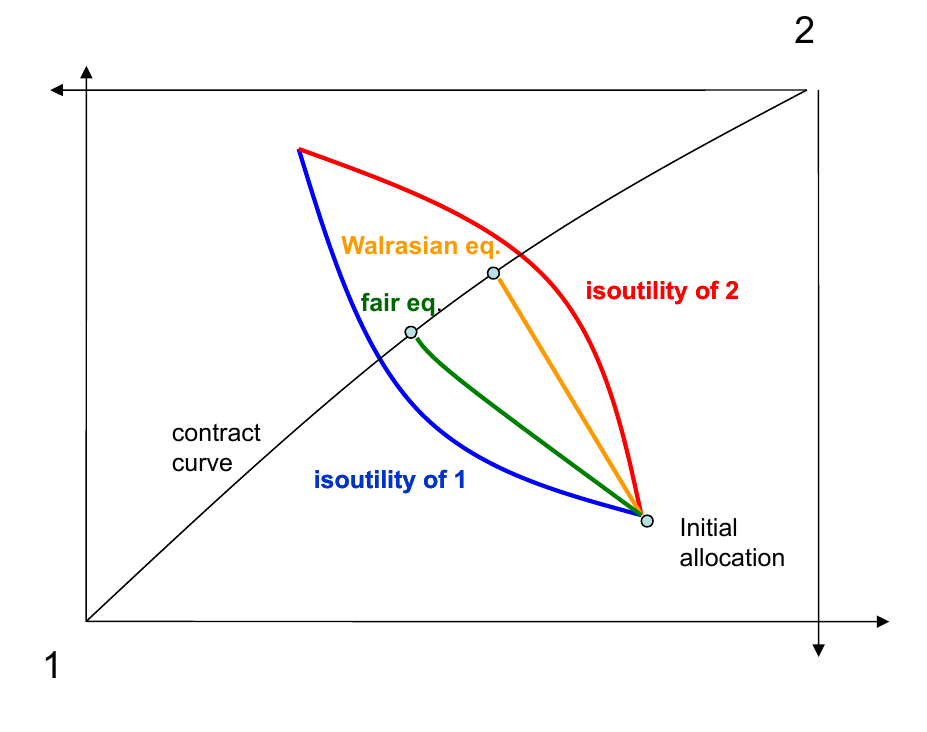}
\includegraphics*[width=7cm]{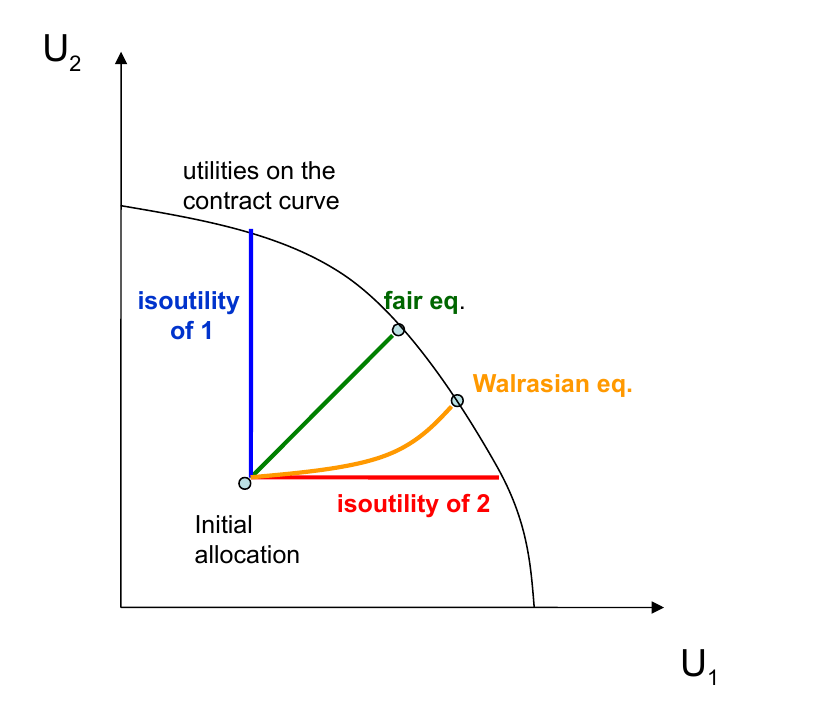}
\caption{Example of the difference between a Walrasian equilibrium and a fair equilibrium in the Edgeworth box and in the space of utilities.}
\label{fig:n=2}
\end{center}
\end{figure}

To have a graphical intuition for our approach, consider Figure \ref{fig:n=2}, where we have $m=2$ (adapted from \citealt{foley}).
In the left panel we represent allocations of the two goods, while in the right panel we represent utilities of the two agents.
The red and the blue lines (in both panels) are the boundaries of the Pareto improving allocations.
The yellow curve is the {Walrasian map from initial endowments to the Walrasian equilibrium allocation: it is a straight line in the Edgeworth box (left), but not necessarily in the space of utilities (right).}
The green line is the path obtained with fair trading: it is a straight line with $45^{\circ}$ inclination in the right panel.
{\color{black} The above example is just illustrative, and we are not claiming that a the limit points of the Fair Trade dynamics cannot coincide with the Walrasian allocation, even if this is generally not true.}

% {\color{red} qui dovevamo definire il walrasian path, e se lo cambiassimo cos\'i­ (da smith-Foley 2008 JEDC di cui aggiungo citazione)? o poi dobbiamo definire walrasian map pure? che dici?}

\section{The network environment}\label{para:network}

What happens if $m \leq n + 1$, if for instance there are only 2 goods and many agents?
In this case we consider a market mechanism based on a \emph{weighted, undirected network} $G$ that allows for distinct couples to match and trade according to the unique fair trading mechanism defined in Section \ref{sec_2fair}.\footnote{%
In Appendix \ref{appendix_more} we show that, instead of assuming a network, we could increase the number of goods, if we want to extend the definition of \emph{fair trading}.}
{
The trade network $G$ is identified by the symmetric adjacency matrix ${\bf W} = (w_{ij})$, where $w_{ii} = 0$ and  $w_{ij} \in [0,1]$ and we assume that the sum of all edge weights in $G$ is equal to 1. The importance of a node $i$ node in terms of the total weight of their connections is given by the \emph{strength},  defined as $s_i = \sum_j w_{ij}$ \citep{Barrat3747}. Note that, given our assumption on the sum of edge weights  $\sum_i s_i = 2$.
The weight of each connection represents the probability that the two agents will trade.}

\noindent
{\color{black} Bilateral trade between agents $i$ and $j$ is given by $f_i(\mu_{i,t}, \mu_{j,t})$ as in Equation \eqref{fair_f}, and trade for agent $i$ on the network $G$ is defined as $f_i^{G} = \sum_{j \in N\setminus\{i\}} w_{ij} f_i(\mu_{i,t}, \mu_{j,t})$, namely the weighted sum of $i$'s bilateral trades with all the agents connected with $i$. The resulting dynamical system is}\footnote{%
Here we consistently define that $f_i \left( {\mu}_{i,t} , {\mu}_{i,t} \right) = {0}$.}
\begin{equation} \label{dinsys0}
\frac{d {\bf X}_t}{d t}  = f^G ({\bf X}(t),{\bf W})
\end{equation}

\noindent 
where:

\begin{equation} \label{dinsys1}
f^G ({\bf X}(t),{\bf W}) = 
 \left(
\sum_{i \in N\setminus{\{1\}}} w_{1i} f_1 \left( {\mu}_{1,t} , {\mu}_{i,t} \right) ,
\sum_{i \in N\setminus{\{2\}}} w_{2i} f_2 \left( {\mu}_{2,t} , {\mu}_{i,t} \right) ,
\dots,
\sum_{i \in N\setminus{\{n\}}} w_{ni}  f_n \left( {\mu}_{n,t} , {\mu}_{i,t} \right)
\right) \ \ .
\end{equation}

\noindent
As for the case of Equation \eqref{fair_mechanism}, this system is not linear.

\smallskip
\begin{proposition}
The fair trading mechanism on a network satisfies zero sum, trade and positive gradient properties.
\end{proposition}
\smallskip

\begin{proof}
Zero sum holds as for every couple $i$ and $j$, which is matched with weight $w_{ij}$,
$f_i=-f_j$ by construction, as discussed in Section \ref{sec_2fair}.

\noindent
Trade property also holds: for every couple $i$ and $j$ such that ${\mu}_i$ and ${\mu}_j$ are linearly independent, consider trader $k$ such that $w_{ik}>0$ and  $w_{jk}>0$  so that both $i$ and $j$ trade with $k$. If ${\mu}_i$ and ${\mu}_j$ are linearly independent, then at least one of them is linearly independent with ${\mu}_k$, suppose it is $\mu_j$.
From fair trading between two agents, as discussed in Section \ref{sec_2fair}, we have that the marginal utility of trader $j$ from that matching is strictly increasing.
Then, as no other trading can generate negative marginal utilities, it means that the overall marginal utility of trader $j$ from all matchings is strictly increasing.
And this can happen only if there is trade, i.e.
\[
{f}^G_j = \sum_{i \in N} w_{ij} f_j \left( {\mu}_{j,t} , {\mu}_{i,t} \right) \ne {0} \ \ .
\]
Finally, positive gradient comes from the fact that $f^G_i$ is a linear combination of $f_i$s, so that
\[
{\mu}_{i,t} \cdot f^G_i
= \sum_{j \in N} w_{ij} {\mu}_{i,t} \cdot  f_i \left( {\mu}_{i,t} , {\mu}_{j,t} \right) \ \ ,
\]
which is strictly positive as long as there is trading.
\end{proof}

\section{A second welfare theorem for networks}

In this section we will prove that there is a one to one mapping between the initial conditions (allocations and network) and the state of the system at any point in time, including in the limit points of the dynamics, where the corresponding allocations are in a subset of the Pareto Set.
Also, we will prove that this map has no holes (is simply connected) and so we can provide a version of the second welfare theorem for networks. We start by providing some definitions and recalling some classical results.

\noindent
\begin{definition}
A function $f: \R^n \to \R^m$ is locally Lipschitz continuous on $\R^n$ if for every $R>0$ there exists a constant $L$ such that:
\begin{equation}
|f(x) - f(y)| \le L |x-y| \, \, \forall x,y \in \R^n \text{ such that }  |x|, |y| \le R
\end{equation}
\end{definition}

\noindent
Recall that resources are fixed in the economy at a point $w \in \R^m$, where the $k$-th coordinate is the total quantity of good $k$ in the economy. Given the initial resources, each possible initial allocation ${\bf X}_0$ is in the set  $E=\{ x \in \R_+^{m \times n} : \sum x_i = w_i \}$.

\noindent
\begin{definition}
Define as $\mathcal{W}^G({\bf X}_0)$ the set of limit points of the fair trade dynamics on networks, that is the set of limit points of \eqref{dinsys1} for a given initial condition ${\bf X}_0  \in E$.%namely $ \mathcal{W}^G = \{ {\bf X}^* \in {\bf X}: {\bf X}^* = \lim_{t \to \infty} f^G({\bf X}(t), {\bf W}) \}$
\end{definition} 

\noindent
From section \ref{para: trading} we know that  $\mathcal{W}^G ({\bf X}_0)$ is a subset of the Pareto Set $\mathcal{W}$ defined in equation (\ref{contract_c}).

\begin{lemma}
If a function is $C^1$ then it is locally Lipschitz
\end{lemma}

\begin{theorem}\label{C-L}
If $f$ is continuously differentiable, then  there exist a unique solution of the dynamical system $\frac{dx}{dt} = f(x)$ satisfying the initial condition $x(0) = x_0$.
\end{theorem}

\begin{proof}
See for example \cite{hirsch1974differential}.
\end{proof}

\noindent
%In what follows we refer to $f: \R^{nm} \to \R^{nm}$ as the function defining the dynamics of trade in \ref{dinsys}, that is, for each agent $i = {1,\dots,n}$ and for each good $k = {1,\dots,m}$:
In what follows we refer to $f_i^G: \R_+^{n \times m} \to \R^{m}$ as the function defining the dynamics of trade for agent $i$ in equation (\ref{dinsys1}), that is, for each agent $i = {1,\dots,n}$, $f_i^G$ is a vector of $m$ components, and each component, for $k = {1,\dots,m}$ is:

\begin{equation}\label{explform}
f^G_{ik} = \sum_{i \ne j}w_{ij} \Bigl(\mu_{ik} - \frac{\sum_{k=1}^m \mu_{ik}(\mu_{ik} +\mu_{jk})
}{\sum_{k=1}^m (\mu_{ik} +\mu_{jk})^2}(\mu_{ik} +\mu_{jk}) \Bigr)
\end{equation}

%\noindent
%For the sake of readability we drop time dependency and vectorial notation, we refer to the vector of good quantities for each agent simply as $x \in \R^{nm}$ and to the probability vector that generates the network as $p \in \R^n$.

\noindent
Where $\mu_{ik} = \frac{\partial U_i}{\partial x_{k}}(x_1,\dots,x_k)$.
For the sake of readability we drop time dependency and we refer to the matrix of good quantities for each agent simply as ${\bf X} \in \R_+^{n \times m}$ and to the adjacency matrix that identifies the network as ${\bf W} \in \mathcal{A} \subset \R^{n \times n}$, where $\mathcal{A}$ is the set of $n \times n$ symmetric matrices such that each entry is a number between 0 and 1, all the diagonal entries are equal to 0 and the sum of all entries of the matrix is equal to 2.

\noindent
We now show that keeping the initial allocation of goods fixed, for each Pareto optimum of the fair trade dynamics there exist a network configuration that implements it.

\begin{theorem}\label{WT}
Consider the fair trade mechanism of pure exchange on networks where agents have continuously differentiable utility functions. Keeping initial allocations ${\bf X}_0$ constant, any ${\bf X}^*$ in the set of the limit points of the fair trade dynamics $\mathcal{W}^G ({\bf X}_0) \subset \mathcal{W}$  can be reached through a sequence of trades for some weighted network $G$. Moreover, each weighted network $G$ leads to a different limit point.
\end{theorem}

\begin{proof}
The fact that any point ${\bf X}^*$ in the set of limit points of te fair trade dynamics can be reached for some weighted network $G$ follows from the definition of $\mathcal{W}^G ({\bf X}_0)$. To prove that each weighted network $G$ leads to a different limit point in $\mathcal{W}^G ({\bf X}_0)$, we prove that the map between the initial conditions (both endowments and network configuration) and the limit points of the trade dynamics is a homeomorphism, that is one-to-one and onto, continuous and with continuous inverse. In order to do that we first transform the parameter ${\bf W}$ into initial conditions, and then we show that the trade dynamics is Lipschitz continuous with respect to both ${\bf X}$ and ${\bf W}$.
%Lipschitz continuity guarantees existence and uniqueness of solutions by (\ref{C-L}), and uniqueness  implies invertibility. 

\medskip
\noindent
Let us start by noting that the dynamical system identified by equation \eqref{dinsys1} has a unique solution given the initial allocation ${\bf X}_0$, for a fixed network ${\bf W}$. This is because each $f_i^G: \R_+^{n \times m} \to \R^{m}$ is continuously differentiable, as we assumed the utility function to be twice continuously differentiable. This can be easily verified checking each of the $k$ components of $f_i^G$ as per equation \eqref{explform}.  It follows that $f^G({\bf X}, {\bf W})$ is at least $C^1$ and by Theorem \ref{C-L} the dynamical system has unique solution.

\medskip
\noindent
Call the solution map $\phi(t,{\bf X}_0, {\bf W})$, unique for each initial condition ${\bf X}_0 \in E \subset \R_+^{n \times m}$. This map is a homeomorphism, that is continuous, one-to-one and with continuous inverse \citep{hirsch1974differential}. We can transform the parameter ${\bf W}$ into initial conditions by introducing a new variable ${\bf S} \in \mathcal{A} \subset \R^{n \times n}$ and imposing that it does not change in time, so that ${\bf S}(t)= {\bf W} $ for all $t$. The system has now variable $\hat{{\bf X}}_t = [{\bf X}_t, {\bf W}]$ and the initial condition is $\hat{\bf{X}}_0 = [{\bf X}_0, {\bf W}]$  We now show that the function $\hat{f}^G(\hat{{\bf X}})$ is Lipschitz in $\hat{{\bf X}}$, so it has a unique solution given initial condition $\hat{{\bf X}}_0$. 
As  ${f}^G$ is Lipschitz in ${\bf{X}}$, $\hat{f}^G$ is Lipschitz in ${\bf{X}}$, so we just need to show that it is Lipschitz in ${\bf W}$ as well. This is straightforward as $\hat{f}^G$ is linear in ${\bf W}$.

\noindent
So $\hat{f}^G$ is Lipschitz continuous both in ${\bf{X}}$ and in ${\bf W}$, which implies that for each initial condition $\hat{\bf{X}}(t_0) = [{\bf X}_0, {\bf W}]$ there is a unique solution, and given our ODE system is autonomous, distinct solutions never cross. Call $\hat{\phi}(t,{\bf X}_0, {\bf W})$ the solution map of  $\hat{f}^G$, for standard arguments this map is a  homeomorphism \citep{hirsch1974differential}. This implies that, given ${\bf X}_0$, changing the weighted network $G$ we reach a different point in $\mathcal{W}^G ({\bf X}_0)$, which concludes our proof.
\end{proof}

\begin{theorem} \label{cont}
If $f$ is Lipschitz continuous in ${\bf X}, {\bf W}$  then the solution $\phi(t,{\bf X}_0,{\bf W})$ is Lipschitz continuous.
\end{theorem}

\begin{proof}
%Being $f$ continuous in ${\bf W}$, ${\bf X}_0$, the map $({\bf W},{\bf X}_0) \to {\bf X}({\bf W},{\bf X}(t_0))$ is also  continuous, so standard arguments imply that ${\bf X}$ is continuous in ${\bf W},{\bf X}(t_0)$ jointly.
This is a classic result on the continuous dependence of solutions on parameters and initial conditions. See for example \cite{hirsch1974differential}.
\end{proof}

\noindent
Note that both if we change the network $G$ and keep the initial endowments ${\bf X}_0$ fixed, and if we keep the network $G$ fixed and we redistribute initial endowments, the limit points of the dynamic will converge to a point that is in the Pareto set. Obviously, given that the solutions are unique, changing the network or the initial allocations (or both), we will reach distinct points in the Pareto set. In our model we can redistribute initial allocations, or redistribute network connections, or both: in every case the limit allocation is Pareto efficient.

\noindent
{This result can be related to the neutrality theorem proved by \cite{Cornet1983}. The author studies a planning procedure, defined as a dynamics over the space of the admissible allocations of an economy with $n$ agents, controlled by a fixed parameter in the simplex of $\R^n$. Provided that the solution never leaves the space of the feasible allocations and that the utility of the agent is non-decreasing along the procedures, \cite{Cornet1983} proves that every point preferred or indifferent by every agent to the inital allocation can be reached by an appropriate choice of the parameter, so that the planning procedure is neutral in the sense defined by \cite{Champsaur1977}. While the results of Cornet's paper are more general, we can stress some connection with our work: we have a system of ordinary differential equation defined by continuous, non decreasing adjustement functions on a compact set, which limit points are a subset of the Pareto set. The network in our model has a similar role to the parameter that attributes a weight to the agents in Cornet: if an agent is disconnected, her utility does not change during the process, the same that happens to an agent with zero weight in Cornet's planning procedure. Finally our Theorem 3 shows that any point preferred or indifferent by every agent to the initial allocation can be reached \emph{via} the appropriate choice of the network on which exchange takes place.}

\begin{lemma}
The set of the limit point of the dynamics $\W^G({\bf X}_0)$ is simply connected.
\end{lemma}
%The solution map $\hat{\phi}(t,{\bf X}_0, {\bf W}): E \times \mathcal{A} \rightarrow\W^G({\bf X}_0) \times {\bf W}$ is simply connected.

\begin{proof}
Consider the solution map ${\phi}(t,{\bf X}_0, {\bf W}): E \times \mathcal{A} \rightarrow\W^G({\bf X}_0) \times {\bf W}$. $E \times \mathcal{A}$ is a convex subset of $\R^{n \times m + n \times n}$ as a product of two convex subset of $\R_+^{n \times m}$ and $\R^{n \times n}$ respectively, so $E \times \mathcal{A}$ is simply connected. Being $E \times \mathcal{A}$ and $\mathcal{W}^G({\bf X}_0) \times {\bf W}$ homoeomorphic,  given $E \times \mathcal{A}$ is simply connected a this is a necessary and sufficient condition for $\mathcal{W}^G({\bf X}_0) \times {\bf W}$ to be simply connected, and so $\mathcal{W}^G({\bf X}_0)$ is simply connected
\end{proof}

%\noindent
%We proved that there is a continuous, invertible with continuous inverse map between the set of initial conditions (initial allocations and network) and that this map has no holes (it is simply connected). Note also that, as we did in the proof of Theorem \ref{WT},  transforming parameters into initial conditions, we can transform initial conditions into parameters, so that we can study the effect of the network topology varying the initial allocations.

\noindent
We can characterise the set of limit points of the fair trading dynamics for any ${\bf X}_0 \in E$:

\begin{proposition}
\label{prop_hom}
The set of the limit points of a fair trading dynamics on networks, $\W^G({\bf X}_0)$, for any ${\bf X}_0 \in E$, is a subset of the Pareto Set  $\W$ which is homeomorphic to a closed $(n-1)$ simplex.
\end{proposition}

\begin{proof}
$\W^G({\bf X}_0)$ is a strict subset of $\W$ as the stable point of the trading dynamics are Pareto Optima and all those allocations in $\W$ where agents are worse off than their initial allocation in the dynamics are not in $\W^G({\bf X}_0)$. Recall that the solution map  $\phi(t,{\bf X}_0, {\bf W})$ is continuous in both ${\bf X}_0$ and  ${\bf W}$. Given that the set $\W$ is homeomorphic to a $(n-1)$ simplex and that $\W^G({\bf X}_0)$ is simply connected, $\W^G({\bf X}_0)$ is also homeomorphic to a $(n-1)$ simplex.
\end{proof}

%\noindent
%The assumption of Lipschitz continuity is central in our context in order to ensure existence and uniqueness of the solution. Note that although sufficient, Lipschitz continuity is not necessary for the existence of a solution continuous in the initial conditions, see for example \cite{Henry73}. Lipschitz continuity is a strong form of uniform continuity which puts a condition on the rate of change of the function, or in other words it puts a bound on its first derivatives. In the case of our interests then a function may fail to be Lipschitz close to the boundary of the goods space, that is to say where $x$ is close to zero, where the rate of change of the function $f$ can be very high. {\color{black} We can rule these cases out by properly choosing the utility function and imposing that for all times $t$, if $x_{i}=0$ we restrict $f_i(x)>0$.}
%It is worth noticing that the eventual failure of the Lispchitz assumption would not necessarily invalidate our results where there still exist a unique solution to the trade dynamics.

\bigskip

\begin{figure}[h!]
\centering
\includegraphics*[width=5cm]{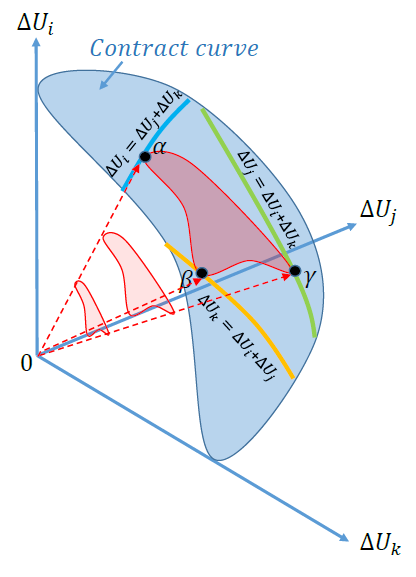}
\caption{Graphical intuition for  Proposition \ref{prop_hom}, representing the projection for three agents of the space of  marginal utilities with respect to the initial allocation.
The light-blue curve represents the Pareto optimal points in $\W$, the red region is the set $\W^G$ of those points that can be reached with our trading mechanism.
If we start with the three distinct stars that have one of those three agents as centers, we will end up in distinct points of $\W$, and hence $\W^G$ has the same dimension of $\W$.}
\label{fig:mani}
\end{figure}

\noindent
It is worth considering an alternative proof of the homeomorphism result in Proposition \ref{prop_hom}, which provides a more intuitive understanding.
Suppose that we have at least three agents, $i$, $j$ and $k$, and that we start from three different star networks: one with  $i$, one with $j$, and one with $k$ in the core.
Unless we start from an allocation that is already Pareto optimal, the three points that we would reach adopting these three networks, and starting from the same initial allocation, cannot coincide.
That is because in a star network, since agents use the fair trading rule, in the limiting point in $\W^G({\bf X}_0)$ the central agent will obtain a marginal utility that is equal to the sum of the marginal utilities of all the other agents.
So, it is impossible that we reach a unique allocation in which half of the overall marginal utilities is given at the same time to each one of the three agents $i$, $j$ and $k$.
Figure \ref{fig:mani} provides a graphical intuition of this argument in the projection of the space of marginal utilities with respect to the initial allocation.

\noindent
The main implication of our result is that we can evaluate the impact of the network structure on the final allocation, as there exists a one-to-one map between each weigted, connected network and the solutions of \ref{dinsys0}. Here we show that, irrespective of initial allocation, agents maximize their utility when they are the core of a star, and that utility in equilibrium is not (always) monotonically increasing in the probability the agent has of being picked to trade.

{\color{black}
\begin{proposition}\label{star_pref}
For any initial allocation of goods any agent strictly prefers to be the core of a star.
\end{proposition}

%{\color{black}
\begin{proof}
{Consider that agent $i$ final utility as a function of the network can be written as:}

\begin{equation}
{
U_i({\bf W}) = U_i(x_i(0)) + \int_0^\infty \sum_{j\ne i} w_{ij} \mu_i(\phi_i(t,{\bf X}_0,{\bf W})) \cdot f_{i,j}(\phi_i(t,{\bf X}_0,{\bf W}), \phi_j(t,{\bf X}_0,{\bf W}))dt}
\end{equation}

\noindent 
{where $\phi_i(t,{\bf X}_0,{\bf W})$ is agent $i$'s solution path as a function of the network. Given that $ \mu_i(\phi_i(t,{\bf X}_0,{\bf W})) \cdot f_{i,j}(\phi_i(t,{\bf X}_0,{\bf W}), \phi_j(t,{\bf X}_0,{\bf W})) \ge 0$ for any $i$ and $j$ along the solution path and that $\sum_j w_{ij} = 1-\sum_{k \ne i}\sum_{j\ne i} w_{kj}$ agent will maximize $U_i({\bf W})$ when $\sum_j w_{ij} = 1$, that is when agent is the core of a star.}

%\begin{equation}
% \frac{\partial U_i({\bf W})}{\partial w_{ij}}=\int_0^\infty  {\mu}_i(x_i(t, {\bf W})) \cdot {f}_{i,j}(x_i(t, {\bf W}))\Biggl [\frac{\partial {\mu}_i(x_i(t, {\bf W}))}{\partial w_{ij}} {f}_{i,j}(x_i(t, {\bf W})) + {\mu}_i(x_i(t, {\bf W}))\frac{\partial {f}_{i,j}(x_i(t, {\bf W}))}{\partial w_{ij}} \Biggr ]
%\end{equation}
%
%and for any $j \ne i$:
%
%\begin{equation}
%\partial \Bigl( \frac{dU_i}{dt} \Bigr)/\partial p_j = \int_0^\infty \frac{1}{n-1} {\mu}_{i,t} \cdot {f}_{i,j}
%\end{equation}
%
%which implies that $ \partial \Bigl( \frac{dU_i}{dt} \Bigr)/\partial p_i  \ge \partial \Bigl( \frac{dU_i}{dt} \Bigr)/\partial p_j$ for any $j$ as ${\mu}_{i,t} \cdot {f}_{i,j}>0$ for any $i$ and $j$.
%So the instantaneous change in utility for agent $i$ is maximised at any time $t$ when $p_i = 1$, which implies that also utility in equilibrium will reach its maximum \emph{ceteris paribus} when $p_i = 1$.
\end{proof}

\section{Effect of network with Cobb-Douglas preferences}\label{simulations}

In this section we investigate \emph{via} numerical simulations how the network structure affects the limit points of the dynamics. We consider Cobb-Douglas preferences over two goods, and we study networks up to 7 nodes. By constructing the contract curve numerically, we provide a graphical example for the 3 nodes network, showing that the stable points of the fair trading mechanism are homeomorphic to a $n-1$ simplex. Moreover, we illustrate the relationship between inequality in network connections and inequality in equilibrium after trade. At the end of the section we provide some evidence on the computational complexity of our process, reporting the convergence times as the number of agents increase.
{As we stressed previously in this paper, we are assuming cardinal utilities, and this is what allows to compute out-of-equilibrium dynamics and to analyse the effect of network metrics on the utility image of the limit points of the dynamics. In interpreting the results of this section, it is important to keep in mind that these are not invariant on monotonic transformations of the utility functions, which is of course a limitation, as ideally we would like to learn how the network structures affect trade \emph{independently} of the specific utility function.}

\medskip
\noindent
As we proved in Theorem \ref{WT}, each network, given an initial point in the commodity space, can be mapped into a solution which, in the limit, converges to a point on the contract curve. Here we provide a graphical illustration of this result.

\noindent
Suppose that agents have a Cobb-Douglas utility function with constant return to scale: $U_i(x)=x_{i,1}^{\alpha_i}x_{i,2}^{1-\alpha_i}$. This implies that the functions are concave, and that the Pareto Set is a curved $(n-1)$ simplex \citep{strat_pareto}.

%\begin{proposition} \label{prop:diff}
%With Cobb-Douglas utility function the set of stable point of a fair trading dynamics is diffeomorphic to a $(n-1)$ simplex.
%\end{proposition}
%
%\begin{proof}
%We know from Theorem (\ref{WT}) that the 
%\end{proof}

\begin{figure}[h]
\centering
\includegraphics*[width=14cm]{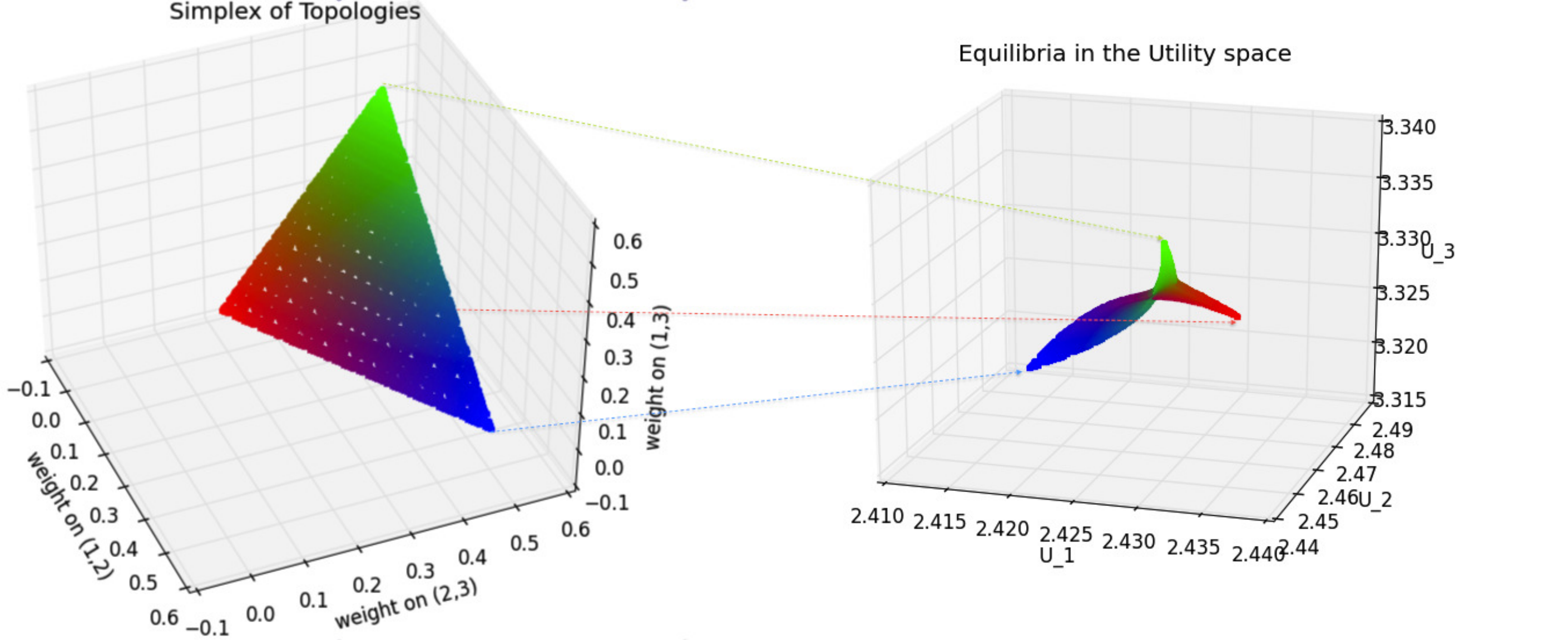}
\caption{Mapping between simplex of representing the space of network configurations and the corresponding equilibria. Only the three vertices are shown, map is according to colours.}
\label{fig:simp}
\end{figure}

\noindent
The leftmost simplex in figure \ref{fig:simp} represents the space of network configurations, each point in that space represents a weighted graph over three nodes: the barycentric coordinates of a point in the simplex correspond to the weight of the network's edges. Each network is then mapped to the corresponding equilibrium of the dynamical system defined by the fair trading mechanism, represented in the space of utilities on the right-hand side of the figure. The map between the two spaces is visualised through colours, a point on the simplex on the left (network) reaches the equilibrium level of utility represented by a point of the same color in the space of utilities. The figure on the right, that is the set of utilities in equilibrium, is a curved 2-simplex, with the vertices of the simplex of networks that are mapped to the vertices of the set of utilities each agent maximizes her utility when is the core of a star, as we showed in Proposition (\ref{star_pref}).

\noindent
In the case represented in  Figure \ref{fig:simp} utility functions are determined by $\alpha_1=0.5$, $\alpha_2=0.4$, $\alpha_3=0.6$, while the initial allocations are such that agent three has the highest endowment of both goods, agent two has  the lowest endowment of good 1 and endowment of good 2 higher than agent 1 that is $x_{3,1} > x_{1,1}>x_{2,1}$ and $x_{3,2}>x_{2,2} >x_{1,2}$. 
%Note that with 3 agents and 2 goods the Edgeworth Box is a simplotope in 4 dimensions, so we cannot visualise it. To visualise the contract curve we can focus on the utility space, that is three--dimensional. Because of theorem  \ref{WT} the set of stable points of the trading dynamics, $\W^N$, is a subset of the contract curve defined in Equation (\ref{contract_c}), that it is homeomorphic to a 2-simplex, and when utility is Cobb-Douglas, we have that the set of stable equilibria is also diffeomorphic to a 2-simplex.

\noindent
The numerical examples provide illustration of our theoretical results: the map between the networks and the set of equilibria  is continuous, and there is a homeomomorphism between the simplex of topologies and the set of equilibria: each initial network is continuously mapped through our dynamical process described in Equation \eqref{dinsys0} into a point of the curved simplex representing the set of limit points of the dynamics.
%Notice also that the points corresponding to the three vertices in the probability simplex, that is to say cases in which the probability of being picked is 1 on one agent and zero on the others, so the cases in which the network is a weighted star with weights 0.5 on each edge, the agent in the core gets the highest gain in utility, as we showed in the previous section.
%The fact that utility is maximised when the agent is the core of a weighted star \emph{ceteris paribus} implies that utility maximising agents would also have preferences on the network structure, so it would be interesting to think of a game where individuals first compete for the position in the network and then trade. 
%Given that every one would prefer to be the only hub of a star, each agent has not only incentive in creating the highest number of connections, but also in destroying the highest number of connections among other agents. 

\medskip
\noindent
{In order to investigate how the network structure affects the equilibrium we compute the trading dynamics for 61367 networks, letting the number of nodes vary from 3 to 7, and we explore the impact of standard network metrics on the utility gain in equilibrium. Of the $2^{n(n-1)/2}$ possible graphs on $n$ nodes we consider simple non-isomorphic ones, both connected and not connected. We include in our computations all simple connected non isomorphic (SCnI) graphs and the subset of simple non-connected isomorphic graphs (SnCnI) such that there are no isolated nodes.  For each network we assume equal edge weights, setting the weight to 1 divided the number of edges.
%Notice that the connected components of $m$ nodes of non-connected non-isomorphic graphs on $n$ nodes with $n>m$ are already included among the connected non-isomorphic graphs on $m$ nodes, so we generate a subset of simple non-connected isomorphic graphs (SnCnI) as the union of SCnI graphs, restricting the attention to cases where the minimum dimension of the smallest connected component is 2 (at least one edge) and excluding trivial cases of disconnected nodes.
As in the section on Pairwise Stability, we assume that all players have the same preferences over two goods $x_1$, $x_2$, represented by a Cobb-Douglas function $U(x_1,x_2) = x_1^{0.5}x_2^{0.5}$ and that initial endowment can be either $e_1 = (1,2)$ or $e_2 = (2,1)$.
Consider that, because of the assumption of homogeneous preferences, if all agents have the same endowment no trade will happen, so we exclude this scenario in our experiments. For each network on $n$ nodes we let the number of agents who have endowment $e_1$ vary from 1 to $n/2$ if $n$ even and $n/2-1$ if odd, and we compute the trade dynamics for all permutations of endowments.
Table \ref{tabnets} reports the number of simple connected non-isomorphic graphs when the number of nodes varies from 3 to 7 and the total number of networks after permuting for initial endowments.}

\begin{table}[htb]
    \centering
\sisetup{table-number-alignment = center, % <-- added/changed
         table-space-text-pre ={(},
         table-space-text-post={\textsuperscript{***}},
         input-open-uncertainty={[},
         input-close-uncertainty={]},
         table-align-text-pre=false,
         table-align-text-post=false}
\begin{threeparttable}
    \caption{Simple connected non-isomorphic graphs}
    \label{tabnets}
\begin{tabular}{c 
                S[table-format=-2.3] % <-- adopted to number of digits in numbers in cells
                S[table-format=-1.4] % <-- adopted ...
                S[table-format=-1.5] % <-- adopted ...
                S[table-format=-1.5]
                S[table-format=-1.5]
                 }
\toprule

   {\# of nodes}       &   {\# of SCnI graphs}   &    {\# of SnCnI graphs}   &   {\# of networks in the experiment}      \\
\midrule
   3       &   2  & 0  & 6 \\
  4           & 6  &  1 &  70            \\
 5       &    21  & 2   & 345    \\
 6         & 112 & 10 & 5002 \\
  7           &  853  & 35   & 55944     \\
\midrule
Total          & 994  & 48  & 61367     \\
\bottomrule
\end{tabular}
\end{threeparttable}
\end{table}

\noindent
{To explore the role of the network structure on the trade dynamics we consider standard network and node metrics, that we define here.
As a measure of the number of connections in the network we use density, which is defined as the number of edges $m$ over the total number of possible edges between $n$ nodes,}
\begin{equation}
{
d = \frac{m}{n(n-1)}
}
\end{equation}

\noindent
{To measure transitivity in a network a common metric is the clustering coefficient \citep{WattsStrogatz1998}, which measures the fraction of triangles over the total number of triads in the network. We adopt a slightly different definition of clustering coefficient, namely the fraction of triangles where one node has a different endowment than the other two over the total number of triads in the network.}
%\begin{equation}
%{
%C = \frac{3 T_{closed}}{T}
%}
%\end{equation}

\noindent
{Given there are two types of endowments, a connection can be either between nodes with the same endowment or between nodes with different endowments. We measure the similarity of connections with respect to initial endowments using assortativity  \citep{Newman2002}, defined as:
}
\begin{equation}
{
r = \frac{\text{Tr} (\boldmath{M})  - || \boldmath{M}^2||}{1- || \boldmath{M}^2||}
}
\end{equation}

\noindent
{where $ \boldmath{M}$ is the mixing matrix of endowments and $ || \boldmath{M}^2 ||$ is the sum of all elements in the matrix $ \boldmath{M}^2 $. Assortativity ranges from -1 (all connections between dissimilar nodes) to 1 (all connections between similar nodes).
}

\noindent
{Under the assumption of homogeneous preferences the number of initial endowments of different type $e_1$ and $e_2$ affects the trading opportunities, so as a control variable we define an endowments similarity index equal to the number of the most scarce endowment type divided the number of the less scarce one. It ranges from 1 (same number of both types of endowment) to 1/6 (highest dissimilarity in the case of 7 agents). 
}

\noindent
{We consider two standard node centrality measures: node strength and betweenness centrality. Node strength is defined as the total weight of node's connections $s_i = \sum_j w_{ij}$ \citep{Barrat3747}, where $w_{ij} \in [0,1]$ is the weight of edge $ij$. Betweenness \citep{Newmanbetw} of node $i$ is defined as the number of shortest paths between pairs of nodes that pass through node $i$ and measures the importance of nodes in connecting different parts of the network.
In addition we construct two further indices: \emph{neighbourhood disassortativity} measures the fraction of neighbours of a node which start with a different endowment than the one of the node, and the \emph{scarcity index} is computed as the fraction of the number of endowments of the same type of node's endowment and the total number of endowments (nodes). The lower the value of the index, the less common the endowment of that agent is and so the more trade opportunity there are for that agent.}

The position in the network determines the trade opportunities each agent has, and as a consequence affects the distribution of the gains from trade in equilibrium. Under the assumption of the fair trade rule, at each instantaneous trade individuals equally split the gain in utility: the star is the most unequal network as the core takes half of the total gain in utility given initial allocations, while nodes in the peripheries only get $1/2(n-1)$ of the total gain in utility, where $n$ is the number of agents. %Consequently inequality in utility gain increases with the number of nodes of the star, which clearly does not necessarily imply that inequality in utility at equilibrium is also increasing.
On the other hand the most equal network is the complete network, where each node trades with each other an equal fraction of time, and where each agent gets $1/n$ of the total gain in utility.

%In general, the fraction of the total gain in utility that each agent gets in equilibrium is equal to half her weighted strength. 
\noindent
We can measure network inequality as the Gini coefficient of the strength (or weighted degree) distribution, using:

\begin{equation}
G_s = \frac{\sum_i^n (2i - n -1)s_i}{n\sum_i^n s_i}
\end{equation}

where $s_i$ is the strength, $n$ is the number of nodes and $i$ is the rank of the strength in ascending order.
The strength distribution of the most unequal network, the star over $n$ nodes, is such that the node in the core (call it $c$) has strength $s_c =1$ and the nodes in the periphery all have strengths  strictly less than 1 and such that $\sum s_{j \ne c} = 1$. %Hence the Gini coefficient is bounded above by $(n-2)/2n$, in the case of a star over $n$ nodes, and bounded below by zero, corresponding to the fully connected network on $n$ nodes where every node has strength $2/n$.
%In this way we can rank every network for its inequality level in terms of utility gains. 

%The total utility gain varies for each network, because it depends on agents' preferences and endowments, hence the dependence between inequality in utility gains and inequality in utility levels at equilibrium is not so straightforward. 

\noindent
Note that higher positional inequality does not necessarily imply higher inequality in utilities after trade.
For example in a network with three agents, ranked according to their initial endowments, it could be that the inequality of the equilibrium distribution of utilities is minimised when we have the poorest agent in the core of a star: if the initial distribution of endowments is highly unequal, stars may promote redistribution, as we will show in the numerical example. 

%\noindent
%We measure the inequality in initial endowments as the Gini coefficient of the sum of the $m$ initial endowments for each agent:
%
%\begin{equation}
%G_e = \frac{\sum_i^n  (2i - n -1)\sum_j^m x_{i,j}}{n\sum_i^n  \sum_j^m x_{i,j}}
%\end{equation}

\noindent
We measure inequality in final utility levels as the Gini coefficient of the individuals' utilities in equilibrium, $U^*_i$:

\begin{equation}
G_u = \frac{\sum_i^n (2i - n -1)U^*_i}{n\sum_i^n U^*_i}
\end{equation}

\noindent
We can use these inequality measures to investigate the relation between positional inequality, endowments inequality and redistribution of welfare, {keeping in mind that $G_u$ is not independent of the specific utility function chosen}. In principle, because of Theorem \ref{WT}, given initial endowments we can find the inequality level in equilibrium for each network configuration, hence a social planner interested in minimising inequality could either redistribute endowments or change the interaction network. Clearly the dependency between the network and inequality in equilibrium is not trivial, as we will show in the numerical exercise.

\begin{table}[h!]
    \centering
\sisetup{table-number-alignment = center, % <-- added/changed
         table-space-text-pre ={(},
         table-space-text-post={\textsuperscript{***}},
         input-open-uncertainty={[},
         input-close-uncertainty={]},
         table-align-text-pre=false,
         table-align-text-post=false}
\begin{threeparttable}
    \caption{Effect of network metrics on utility gain}
    \label{tab1}
\begin{tabular}{r 
                S[table-format=-2.3] % <-- adopted to number of digits in numbers in cells
                S[table-format=-1.4] % <-- adopted ...
                S[table-format=-1.5] % <-- adopted ...
                 }
\toprule
                    & \multicolumn{3}{c}{Aggregate utility gain}                                \\
\midrule
Intercept        &   -0.5030***   & (0.006)     \\

\midrule
Clustering (dissimilar triangles only)        &  0.0231***    & (0.001)            \\
Assortativity   &        -0.0300***         &  (0.001)             \\
Number of nodes  & 0.1030*** & (0.001)\\
Connected & 0.0925\tnote{***} & (0.002)  \\
Endowments similarity  & 0.3635*** & (0.001)\\
\midrule
Observations        & {61367}            &                           \\
R-squared        & 0.818           &         \\
Joint significance (p-value F-statistics) & 0.00 &    \\
\bottomrule
\end{tabular}
    \smallskip
    \footnotesize
standard error in parentheses\par
\begin{tablenotes}[para,flushleft]
    \item[*]    $p < 0.10$,
    \item[**]   $p < 0.05$,
    \item[***]  $p < 0.01$
    \end{tablenotes}\par
\end{threeparttable}
\end{table}

\begin{table}[h!]
    \centering
\sisetup{table-number-alignment = center, % <-- added/changed
         table-space-text-pre ={(},
         table-space-text-post={\textsuperscript{***}},
         input-open-uncertainty={[},
         input-close-uncertainty={]},
         table-align-text-pre=false,
         table-align-text-post=false}
\begin{threeparttable}
    \caption{Effect of network metrics on utility gain}
    \label{tab2}
\begin{tabular}{r 
                S[table-format=-2.3] % <-- adopted to number of digits in numbers in cells
                S[table-format=-1.4] % <-- adopted ...
                S[table-format=-1.5] % <-- adopted ...
                 }
\toprule
                    & \multicolumn{3}{c}{Aggregate utility gain}                                \\
\midrule
Intercept        &   -0.5038***   & (0.007)     \\

\midrule
Density      &   0.0055\tnote{***}            &  (0.001)  \\
Assortativity   &        -0.0309***         &  (0.001)             \\
Number of nodes  & 0.1028*** & (0.001)\\
Connected & 0.0932\tnote{***} & (0.002)  \\
Endowments similarity  & 0.3654*** & (0.001)\\
\midrule
Observations        & {61367}            &                           \\
R-squared        & 0.817           &         \\
Joint significance (p-value F-statistics) & 0.00 &    \\
\bottomrule
\end{tabular}
    \smallskip
    \footnotesize
standard error in parentheses\par
\begin{tablenotes}[para,flushleft]
    \item[*]    $p < 0.10$,
    \item[**]   $p < 0.05$,
    \item[***]  $p < 0.01$
    \end{tablenotes}\par
\end{threeparttable}
\end{table}

\noindent
{We investigate the effect of network metrics on the aggregate gain in utility for the network, that is the sum for all nodes of their gain in utility after trade (utility of equilibrium endowments -  utility of initial endowments ). Tables \ref{tab1} and  \ref{tab2} report the results of the OLS regressions. In both cases we control for the number of nodes, the dissimilarity in initial endowments and wether the network is connected or not. As density and clustering are highly correlated (Pearson's correlation coefficient 0.83) in order to avoid multicollinearity we drop one of the two alternatively. Comparing  tables \ref{tab1} and \ref{tab2} we can see that while both density and clustering have a significative positive impact on utility gain, clustering shows a stronger correlation than density. Both regressions show that a more disassortative network brings higher utility gains, even if the magnitude of the coefficient is less than we expected.}

\begin{figure}[h!]
\centering
\includegraphics[width = 0.5\textwidth]{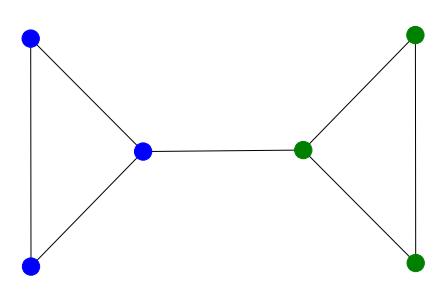}
\caption{Highly assortative networks generating high utility gain}
\label{fig:highass}
\end{figure}

\noindent
{A possible explanation is illustrated in the example in figure \ref{fig:highass}, showing a highly assortative network which generates a high aggregate utility gain: the two groups with different endowments manage to profitably trade thanks to the two agents bridging them, who are going to extract the highest utility gain from trade. So while assortativity has little impact on aggregate utility gain, it has a stronger effect on the distribution of this gain. To see this consider a simple exercise: we take all possible permutations of initial endowments of the network in figure \ref{fig:highass}, and we check the relationship between the Gini coefficient of utility after trading and the assortativity index, controlling for endowments similarity. Table \ref{tab3} shows that the more assortative the network the higher the inequality post-trade, and we can check that most of the variance in inequality is explained by assortativity. This is not necessarily true if the network is more densly connected and there are no nodes which have a clear advantage because of their position. See the appendix for more details.}

\begin{table}[h!]
    \centering
\sisetup{table-number-alignment = center, % <-- added/changed
         table-space-text-pre ={(},
         table-space-text-post={\textsuperscript{***}},
         input-open-uncertainty={[},
         input-close-uncertainty={]},
         table-align-text-pre=false,
         table-align-text-post=false}
\begin{threeparttable}
    \caption{Effect of network assortativity on inequality}
    \label{tab3}
\begin{tabular}{r 
                S[table-format=-2.3] % <-- adopted to number of digits in numbers in cells
                S[table-format=-1.4] % <-- adopted ...
                S[table-format=-1.5] % <-- adopted ...
                 }
\toprule
                    & \multicolumn{3}{c}{Gini of post-trade utility}                                \\
\midrule
Intercept        &   0.0150***   & (0.001)     \\

\midrule
Assortativity   &        0.0185***         &  (0.001)             \\
Gini scarcity index  & 0.0283*** & (0.007)\\
\midrule
Observations        & {41}            &                           \\
R-squared        & 0.828           &         \\
Joint significance (p-value F-statistics) & 0.00 &    \\
\bottomrule
\end{tabular}
    \smallskip
    \footnotesize
standard error in parentheses\par
\begin{tablenotes}[para,flushleft]
    \item[*]    $p < 0.10$,
    \item[**]   $p < 0.05$,
    \item[***]  $p < 0.01$
    \end{tablenotes}\par
\end{threeparttable}
\end{table}

\noindent
The relation between the networks and inequality has been explored in \citep{Hidalgo14}, who find a relation between the network structure and meritocracy: when the network is sparse then individuals' compensations depend on the position in the network instead of their ability to produce value. In a different setting \citep{bowles11}  study the impact of networks on inequality where agent play a coalitional game. They find a connection between network sparseness and inequality by studying how the extremal Lorenz distribution changes under different networks.
{We investigate the impact of the network on the distribution of welfare at equilibrium estimating the dependence of inequality in post-trade utility on positional inequality, measured as the Gini coefficient of strength distribution. We include in the OLS regression assortativity and density and we control for the Gini coefficient of the scarcity index. Table \ref{tab4} shows that higher Gini index of strength is associated with higher inequality in utility post-trade, and the effect is significative and quite strong. More assortative networks also lead to a more unequal distribution of utility after trade, as we illustrated with the example of the 6-node network above, while on the contrary a more dense network leads to more equal final distribution. It is important to note that in our experiment there never is high inequality in terms of initial endowments: there are only two possible endowments which have the same value in utility term for all agents, and an agent can have an initial advantage only if they own a relatively scarce endowment. The magnitude of this is not large, as the maximum value of the Gini coefficient of the scarcity index is 0.15. This is fundamental to interpret the result of our experiment: when inequality in initial endowments is negligible, inequality in connection is the most important factor affecting distribution after trade. In the appendix we show what happens if we allow for large inequality in initial endowments.}

\begin{remark}
{
When there are low levels of inequality in inital endowments, a more equal allocation in equilibrium can be implemented by redistributing network strength from high strength agents to low strength ones.}
\end{remark}

\begin{table}[h!]
    \centering
\sisetup{table-number-alignment = center, % <-- added/changed
         table-space-text-pre ={(},
         table-space-text-post={\textsuperscript{***}},
         input-open-uncertainty={[},
         input-close-uncertainty={]},
         table-align-text-pre=false,
         table-align-text-post=false}
\begin{threeparttable}
    \caption{Effect of network metrics on inequality}
    \label{tab4}
\begin{tabular}{r 
                S[table-format=-2.3] % <-- adopted to number of digits in numbers in cells
                S[table-format=-1.4] % <-- adopted ...
                S[table-format=-1.5] % <-- adopted ...
                 }
\toprule
                    & \multicolumn{3}{c}{Gini of post-trade utility}                                \\
\midrule
Intercept        &   0.099***   & (0.000)     \\

\midrule
Gini strength & 0.0254*** & (0.000) \\
Assortativity   &        0.0087***         &  (0.000)             \\
Density & -0.0084*** & (0.000) \\ 
Gini scarcity index  & 0.0458*** & (0.001)\\
\midrule
Observations        & {61367}            &                           \\
R-squared        & 0.615           &         \\
Joint significance (p-value F-statistics) & 0.00 &    \\
\bottomrule
\end{tabular}
    \smallskip
    \footnotesize
standard error in parentheses\par
\begin{tablenotes}[para,flushleft]
    \item[*]    $p < 0.10$,
    \item[**]   $p < 0.05$,
    \item[***]  $p < 0.01$
    \end{tablenotes}\par
\end{threeparttable}
\end{table}

\noindent
{To see how the position of an agent in the network affects her own utility, we regress the utility gain from trade on node strength and node betweenness centrality, controlling for how assortative the immediate neighbourhood of the node is and how scarce is the endowment of the agent. The results of the OLS regression are reported in table \ref{tab5}. Node strength and betweenness centrality  are both positively correlated with the utility gain in equilibrium; the higher the fraction of agent's neighbours, the larger the utility gain. Moreover agents endowed with larger quantities of the relatively scarce good in the economy are able to extract more utility from trade. All effects are significative.}

\begin{table}[htb]
    \centering
\sisetup{table-number-alignment = center, % <-- added/changed
         table-space-text-pre ={(},
         table-space-text-post={\textsuperscript{***}},
         input-open-uncertainty={[},
         input-close-uncertainty={]},
         table-align-text-pre=false,
         table-align-text-post=false}
\begin{threeparttable}
    \caption{Effect of node metrics on post-trade utility}
    \label{tab5}
\begin{tabular}{r 
                S[table-format=-2.3] % <-- adopted to number of digits in numbers in cells
                S[table-format=-1.4] % <-- adopted ...
                S[table-format=-1.5] % <-- adopted ...
                 }
\toprule
                    & \multicolumn{3}{c}{Individual utility gain}                                \\
\midrule
Intercept        &   0.0505***   & (0.000)     \\

\midrule
Strength        &  0.1246***    & (0.000)            \\
Betweenness      &   0.0702\tnote{***}            &  (0.000)  \\
Neighbourhood disassortativity   &        0.0583***         &  (0.000)             \\
Scarcity index  & - 0.0875*** & (0.000)\\
\midrule
Observations        & {423643}            &                           \\
R-squared        & 0.832           &         \\
Joint significance (p-value F-statistics) & 0.00 &    \\
\bottomrule
\end{tabular}
    \smallskip
    \footnotesize
standard error in parentheses\par
\begin{tablenotes}[para,flushleft]
    \item[*]    $p < 0.10$,
    \item[**]   $p < 0.05$,
    \item[***]  $p < 0.01$
    \end{tablenotes}\par
\end{threeparttable}
\end{table}

\bigskip
\noindent
\cite{Axtell05} proves that  {for decentralized exchange processes  where groups of agents trade, provided that trade is individually rational so that the sum of utilities increases monotonically as long as there is trade, computational complexity is P (the number of interactions is bounded above by a polynomial of the number of agents and commodities). Moreover, analyzing the case of individually rational bilateral exchanges where couples of agents with Cobb-Douglas preferences are randomly matched to trade, \cite{Axtell05} finds that the number of interactions required to reach convergence to the equilibrium is linear in the number of agents, and increasing the number of commodities just increases the number of interactions needed without changing the linear dependency with the number of agents. Our model can be seen as an instance of this type of bilateral exchange, where instead of a random pairing of agents we have a probability distribution on the couples, represented by a weighted network with sum of weights equal to 1. On the basis of this result, if we consider our model of decentralized exchange with Cobb-Douglas utility, we expect a linear relationship between the number of agents and the number of interactions required for convergence. In order to illustrate this} we computed the convergence times for our process with 2 goods and homogeneous Cobb-Douglas utility functions, letting the number of agents vary from 3 to 100. For each experiment initial allocations of endowments were randomly chosen, and the network considered is a star network with a random agent in the core.
Each process is stopped at step $T$ if the difference between the amount of goods that each agent has at $T$ and $T-1$ is less than $\epsilon = 0.00001$.
As figure (\ref{fig:compl}) shows, the relation between the number of agents and the number of interactions needed for convergence appears to be linear.
{Notice that this linear relationship is independent of the network, as it holds for any probability distribution over couples of agents, so for any weighted network where the sum of weights is one. To conclude, based on \citep{Axtell05} we can affirm that the complexity of our exchange process is P, and that in the case of Cobb-Douglas preferences the number of interactions is linear in the number of agents. While do not compute convergence times for other specification of preferences, \cite{Axtell05} affirms that the linear relationship holds for CES utilities as well, and we would expect this to hold for our model as well.}

%\begin{equation}
%
%\end{equation}

\begin{figure}[H]
\centering
\includegraphics[width = 0.8\textwidth]{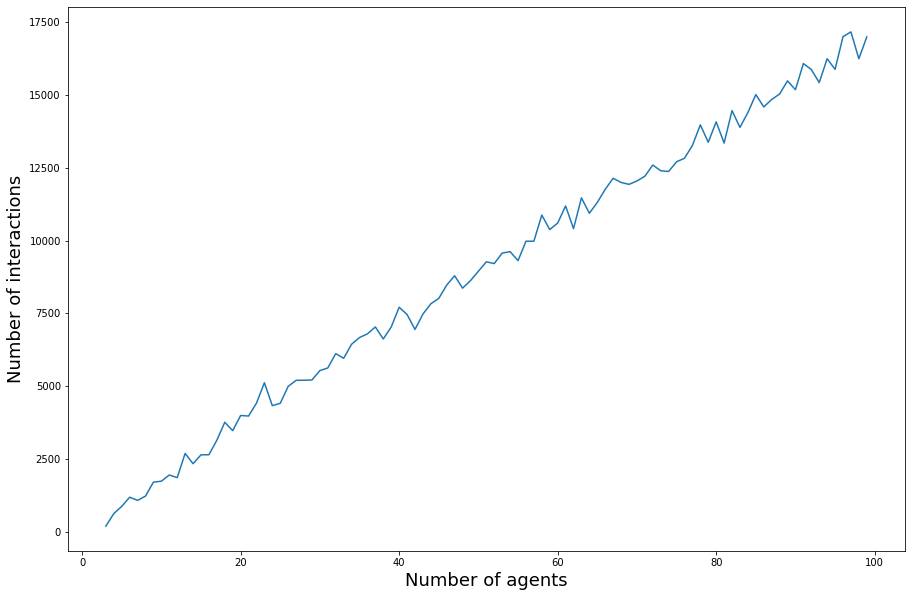}
\caption{Number of interactions required for convergence as a function of the number of agents, termination when $\epsilon = 0.00001$} 
\label{fig:compl}
\end{figure}

\section{Pairwise Stability}
{
In this section we provide some insights about the Pairwise Stability of trade networks.
Following  \cite{JacksonWolinsky96}, we define the \emph{value} of a trading network as the sum of individual utilities \emph{after trading}, $v(g) = \sum_i u_i (g) $. The \emph{allocation rule} is such that any agent gets the utility of her equilibrium allocation minus the cost of her links, so the payoff of agent $i$ trading on network $g$ is defined as:}
\begin{equation}
{
Y_i(g,v) =   u(x_i^*) - \sum_{j:ij \in g} c_{ij}
}
 \end{equation}
 {
where $x_i^*$ is the equilibrium allocation for agent $i$ after trading on $g$ and $c_{ij}$ is the cost of link $ij$.}
\begin{definition}
{
The network $g$ is pairwise stable if}

\begin{enumerate}
\item {for all $ij \in g$, $Y_i(g,v) \ge Y_i(g-ij,v)$ and $Y_j(g,v) \ge Y_j(g-ij,v)$}
\item {for all $ij \notin g$, if $Y_i(g,v) < Y_i(g+ij,v)$ then $Y_j(g,v) > Y_j(g+ij,v)$}
\end{enumerate}
\end{definition}

{
As our networks are weighted we need to make some assumptions about how we add and severe links. We construct our rule such that the sum of weights in the network is preserved: when an edge is added (severed) all remaining weights are decreased (increased) by the same proportion, such that the sum of all weights is one. As the network we analyse have equal weights on all edges, the new weights are simply $w_a = \frac{1}{|E|+1}$  and $w_d = \frac{1}{|E|-1}$ when we add and severe an edge respectively, where $|E|$ is the number of edges in the original network.}

{
As an example, consider a 4 nodes star, where all edges have the same weight, 1/3. After adding an edge between two peripheral nodes all edges will have weight 1/4. When we severe an edge from the star, the remaining two edges will have weight 1/2.}

{
As the trade opportunities and consequently the payoff of each player depend on the initial allocations as well as on the network, it is not possible to analyze the stability of trading networks separately from the initial allocations. To investigate pairwise stability of trading networks we adapt the trading example of \cite{JacksonWatts98}, and we restrict our attention to four types of networks:
}

\begin{enumerate}
\item  {star with equal weight $w_s=\frac{1}{n-1}$ on all edge.s}
\item  {multi-hub star, where the $k$ agents in the core are connected to everybody else and all edges have weight $w_k = \frac{1}{k(n-k)+\binom{k}{2}}$}
\item  {circle with equal weight $w_o = \frac{1}{n}$ on all edges.}
\item  {complete network with equal weight $w_c = \frac{1}{\binom{n}{2}}$ on all edges.}
\end{enumerate}

{
For the sake of simplicity we consider that there are two goods in the economy, and all agent have the same utility function, a symmetric Cobb-Douglas $U(x) = x_1^{0.5}x_2^{0.5}$. Moreover we assume that the cost of creating and maintaining a link is the same for all edges and equal to $c$. An agent's endowment is either (2,1) or (1,2), and we consider all possible combinations of endowments, provided that at least one of the endowments is different (otherwise there would be no trade given our assumptions on the utilities). Considering the total endowment of each good in the economy, that is the sum of that good for all agents, we define as \emph{scarce} the good which total endowment is less, if this good exists. Note that as long as the number of agents is odd, there always is a scarce good in the economy, while if the number of agents is even, the only case in which there is no scarce good is when half of the agent have endowment (2,1) and the other half (1,2). The results that follow are obtained computing the fair trade dynamics for each of the four types of networks listed above, for any combinations of initial endowments such that at least one is different. For each of these networks then we compute the fair trade dynamics of all networks obtained adding and severing links according to the rule described above. Finally, we compare the utility of the equilibrium allocation net of the cost of links of the initial network with that of each network obtained from it by adding or severing links, hence determining Pairwise Stability.}

{
Note that in all of the following results, if a network is pairwise stable for a positive link cost $c>0$ the upper (and lower) bound for $c$ for which stability holds will depend on the extra utility that agents gain after trading, which ultimately depends on the network itself, the initial allocations and agent preferences. We call the upper bound $\bar{c}_g^n$ and the lower bound $\ubar{c}_g^n$  where $g = \{s, k, o, c\}$ stands for star, multi-hub with $k$ nodes in the core, circle and complete network respectively and $n$ is the number of nodes. %(Questa la aggiungiamo o Ã¨ ridondante come notazione?)
}

\begin{result}
 {
The star network is pairwise stable  for some $0 \le c \le \bar{c}_s^n$ only if there is one scarce good and only one agent owning more of the scarce good in her initial endowment, and that agent is in the core.
}
\end{result}

{
The interpretation of this result is straightforward: given there is only one agent possessing the scarce good, every other agent wants to trade with her and only with her. So there is no advantage in creating a link between any pair of nodes in the periphery and no agents in the periphery has any incentive in deleting the link with the core. If there are two (or more) agents owning the scarce good, then the star is not pairwise stable: if one of them is in the core, while agents in the periphery have no incentive to delete an existing link with the core, they all have incentive in linking with the other agent owning the scarce good. This agent is indifferent between deleting the link with the core or keeping it, but is strictly better off linking with all other agents.
} 

\begin{result}
 {
The multi-hub network with $k$ nodes in the core is pairwise stable for $c =0$ }
\begin{enumerate}
\item  {if there is one scarce good and there are at most $k$ agents who have more of the scarce good in their initial endowments and they are all in the core.}
\item  {if there is no scarce good, $n$ is even, $k=n/2$, all agents in the core have the same endowments and all agents in the periphery have the same endowments.}
\end{enumerate}
\end{result}

{
Once all agents who do not own the scarce good are connected with the agents owning the scarce good, they have no incentive to delete nor to add any other link. Agents owning the scarce good will be indifferent between severing the link between them or not (if the cost is zero), but will strictly prefer not to severe any link with any of the peripheral agents not owning the scarce good. Note that if there are $k$ agents who own the scarce good and they are all connected with agents not owning the scarce good but not between them, the network is pairwise stable for a non negative cost, $0 \le c \le \bar{c}_k^n$. This is because links between agents with the same endowment are worthless.
}

\begin{result}
 {The circle network}
\begin{enumerate}
\item  {Is pairwise stable for $n \ge 4$ and $n$ even,  $ \ubar{c}_o^n \le c \le  \bar{c}_o^n$,  there is no scarce good and any couple of neighbours has different initial endowments.}
\item {Is not pairwise stable for any $c\ge0$ for $n>4$ and $n$ odd.}
\end{enumerate}
\end{result}

{
Consider the circle with 4 agents}\footnote{In the case $n=3$ the circle network and the complete network are indistinguishable}{, no scarce goods and all couple of connected agents having different initial endowments. No agent has incentive to severe links as long as the gain they obtain from trading with their neighbours is more than the cost of these links. There are only two links that can be added in this case, and they are both between agents with the same initial endowments, so no profitable trade between them, hence no incentive to form them. Networks with $n>4$ and $n$ odd are never pairwise stable as in this case at least one link is between agents with the same endowment: they will have incentive to severe that link and connect to agents with different endowments. When $n>4$ and $n$ even, provided that any couple of neighbours have different endowments, the circle is pairwise stable only for a positive link cost $c$. This is because for each agent there are $\frac{n}{2}-2$ other agents with different endowments which are not in the agent's neighbourhood. This means that these are all profitable links which both agents will be incentivized to form, unless the cost of these links is higher than the extra gain they can get from forming them.
}

\begin{result} {
The complete network is pairwise stable for $c=0$ both if there is a scarce good and if there is no scarce good.}
\end{result}

{
First we can notice that the complete networks are stable only if $c=0$. This is because in a complete network there will be worthless links between agents which have the same initial endowments, so as long as the cost of linking is positive agents will be better off severing those links.  Connections with neighbours with different endowments are valuable, so there is no incentive in severing them.
}

\begin{result}  {
Any network $g$ where all and only agents with different endowments are connected is pairwise stable for some $0\le c \le \bar{c}_g^n$. If all connections are between agents with different endowments but not all of them are connected, then the network is pairwise stable for a strictly positive $\ubar{c}_g^n \le c \le \bar{c}_g^n$.}
\end{result}
%Add: general result on redistribution: which network maximizes the reduction in inequality (the star)
%Is there a trade-off with efficiency? yes. Can I compute the efficiency loss? maybe.

\section{Conclusions}

This paper studies an Edgeworth process on weighted graphs, where agents can continuously exchange their endowments with their neighbours, driven by their utility functions. {Under the assumption of cardinal utilities,} we define a family of trade dynamics which fixed points coincide with the Pareto Set, and choose a specific mechanism in this family, according to which individuals equally split the utility gain of every trade. This choice is without loss of generality as the results obtained hold for all trade mechanisms that satisfy \emph{zero sum, trade} and \emph{positive gradient}. Under usual assumptions on the structure of preferences we prove a version of the Second Welfare Theorem on networks: for any weighted connected network, there exists a sequence of Pareto improving trades which ends in the Pareto Set. 
Assuming Cobb-Douglas preferences, we build numerical examples of the mapping between the network topology and the final allocation in the Pareto Set, and provide a brief analysis of the impact of the topology on the final allocation. We believe that the relationship between the network and inequality should be further analysed, to understand the link between deprivation in endowments and deprivation in opportunities determined by the position on the network.

%\FloatBarrier
\bibliographystyle{chicago}
\bibliography{references}

\appendix
%dummy comment inserted by tex2lyx to ensure that this paragraph is not empty
\global\long\def\thesection{\Alph{section}}
 \global\long\def\thesubsection{\Alph{section}.\arabic{subsection}}
 \setcounter{proposition}{0} \global\long\def\theproposition{\Alph{proposition}}
 \setcounter{definition}{0} \global\long\def\thedefinition{\Alph{definition}}
 \setcounter{figure}{0} \global\long\def\thefigure{\Alph{figure}}
 \global\long\def\theexample{\Alph{example}}
 \setcounter{example}{0}

\section{More agents} \label{appendix_more}

{Here we build on the definition of \emph{fair trading} in Section \ref{sec_2fair}, to show that, if there are more agents, and if every agent can trade with anyone else, we need to increase the number of goods if we want to extend the definition.} 

Suppose now that there are more than two agents, so that $n \geq 3$.
Trade is always bilateral, and \emph{fair trading} implies that for every trade the marginal utility from trading has to be  equally split among the parts:
\bge \label{ort_all}
\left( {\mu}_{i,t} + {\mu}_{j,t} \right) \cdot f_i \left( {\mu}_{i,t} , {\mu}_{j,t} \right) = 0  \quad \forall i,j  \in N, i \ne j\ \
\ene 

\noindent
This must hold for all of the $n-1$ possible couples where trader $i$ is involved, so that individual $i$'s instantaneous trade ${f}_i$ lies in a sub--space of dimension $m-n+1$, if it exists.
This clearly imposes a first constraint on the minimal possible amount $m$ of goods.

\noindent
Moreover, by the zero sum property, we need that the sum of all the istantaneous trades cancels out,  $\sum {f}_i =0$.
This is an additional constraint, that will be satisfied only if the dimension of the sub--space where ${f}_i$ lies is more than one.
So the minimum number of goods that guarantees the existence of fair trading is such that $m-n+1 \geq 2$, or that $m \geq n+1$.

\begin{proposition}
If $n \ge 3$ then fair trading mechanism exists if and only if $m \ge n+1$
\end{proposition}

\begin{example}
[3 traders]
Suppose that for a certain allocation all the three vectors of marginal utilities of the traders are linearly independent.
Say ${\mu}_1 = (2,1,1)$, ${\mu}_2 = (1,2,1)$ and ${\mu}_3 = (1,1,2)$.
${f}_1$ has to be orthogonal to both ${\mu}_1 + {\mu}_2 = (3,3,2)$ and ${\mu}_1 + {\mu}_3 = (3,2,3)$, so that it will be of the form ${f}_1 = k(5,-3,-3)$, for some $k \in \R$.
Similarly we will have ${f}_2 = h(-3,5,-3)$, for some $h \in \R$, and ${f}_3 = \ell (-3,-3,5)$, for some $\ell \in \R$. \\
To balance trading we need also that ${f}_1 + {f}_2 + {f}_3 = (0,0,0)$, but as they are linearly independent vectors, this is possible only for $k=h=\ell=0$, which means no trading, even if marginal utilities are not proportional.
\hfill $\square$
\end{example}

\begin{remark}
If the fair trading is between two traders $(n=2)$ then two goods $(m\ge 2 )$ are sufficient to guarantee the existence of trade
\end{remark}

\noindent
The above can be easily verified, with two traders each trade ${f}_1$ and ${f}_2$ by construction is  orthogonal to the same vector ${f}_1+{f}_2$, so that they will never be linearly independent.

\noindent
The previous example shows that if $m \leq n$, and $m \geq 3$, then fair trading is not possible.
If the number of goods where instead $m = n+1$, then every candidate ${f}_i$ would lie on a plane, and there would always exist a non--trivial solution for the zero sum property because we would have a homogeneous system of linear equations with $n$ linear equations in $n+1$ variables.
If $m$ is even greater, then existence would result \emph{a fortiori}.

\section{Experiment with large inequality in intial endowments}

\noindent
If we allow inequality in initial endowments to be larger we expect that when agents who are disadvantaged in endowments are also disadvantaged in terms of network connections, inequality in network strength will increase inequality in utilities, and viceversa.
To verify this hypothesis we consider a different exercise, allowing much larger variation in endowments and edge weights. To limit the computational burden, instead of considering all simple non isomorphic graphs we restrict our attention to a specific family of weighted connected graphs that can be seen as a linear combination of stars, where the minimally connected network is a star and the maximally connected network is a complete one. This class of networks is a weighted analogous to \emph{nested split graphs} \citep{konig2014nestedness}, that are graphs with a nested neighbourhood structure, where the set of neighbours of lower degree nodes is contained in the set of neighbours of higher degree ones.
Except for the limiting case of the complete graph, the nodes in our networks can be divided in two partitions according to their degree: nodes in the core are connected between each other and with all the nodes in the periphery, while nodes in the periphery are connected only to nodes in the core, giving a \emph{multi-hub} network.
We generate networks in this class with 3, 5 and 7 nodes and we let initial endowments of the two goods in the economy vary such that the total quantity of each good is constant across all experiments, $e_1 = \sum_n x_{i,1} = 30$ and $e_2 = \sum_n x_{i,2} = 18$ respectively. All agents have the same preferences over the two goods, represented by a Cobb-Douglas function $U(x_1,x_2) = x_1^{0.5}x_2^{0.5}$.
To initialise each experiment with $n$ agents, we generate a set of different initial endowments such that the sum of good 1 is 30 and the sum of good 2 is 18 and we compute the limit points of the trading dynamics for each network and endowments, generating 107484 experiments with 3 agents, 76433 experiments with 5 agents and 10762 with 7.
We then split the obtained dataset according to the following rule: we rank agents in terms of initial endowments and of network strength: the agent with the largest endowments gets the highest endowments ranking and the agent with the largest strength gets the highest strength ranking. The ranking for initial endowments is found by evaluating utility at initial conditions for each agent\footnote{Alternative ranking measures have been evaluated, namely the sum of initial endowments and the Euclidean norm of the vector of initial endowments, and they do not change our results.}.
Then we put in one group, (\emph{different rank}) all those experiments for which the most disadvantaged agent in terms of endowments is advantaged in terms of connections. This means all the experiments where the poorest agent is at most second in the network strength ranking with 3 agents, first with 5 and 7 agents. Conversely we collect in the other group (\emph{similar rank}) the remaining experiments.
Results are summarized in Table  \ref{tabaa}, showing a positive and significant relationship between inequality in endowments and inequality in equilibrium utility in the \emph{same rank} group and a negative significant relationship in the \emph{different rank} both with homogeneous and non homogeneous preferences. On the basis of this result we can say:

\begin{table}[htb]
    \centering
\sisetup{table-number-alignment = center, % <-- added/changed
         table-space-text-pre ={(},
         table-space-text-post={\textsuperscript{***}},
         input-open-uncertainty={[},
         input-close-uncertainty={]},
         table-align-text-pre=false,
         table-align-text-post=false}
\begin{threeparttable}
    \caption{Inequality dependence on endowments and strength}
    \label{tabaa}
\begin{tabular}{r 
                S[table-format=-2.3] % <-- adopted to number of digits in numbers in cells
                S[table-format=-1.4] % <-- adopted ...
                S[table-format=-1.5] % <-- adopted ...
                S[table-format=-1.5]
                S[table-format=-1.5]
                 }
\toprule
               
       OLS             &   {Same Rank}       &   {Different Rank}          \\
\midrule
Intercept        &   0.0186***       &    0.0163***            \\
                    & (0.000)           & (0.000)                     \\

\midrule
Gini Endowments        &  0.9521***       &    0.9604***         \\
                    & (0.001)           & (0.001)           \\
\addlinespace
Gini strength      &   0.0510***           &   -0.0072***          \\
                    &  (0.002)           &  (0.003)              \\
\midrule
Observations        & {99974}            &   {72397}             \\
R-squared        & 0.821          &   0.845          \\
Joint significance& 0.00 &  0.00    \\
\bottomrule
\end{tabular}
    \smallskip
    \footnotesize
standard error in parentheses\par
\begin{tablenotes}[para,flushleft]
    \item[*]    $p < 0.10$,
    \item[**]   $p < 0.05$,
    \item[***]  $p < 0.01$
    \end{tablenotes}\par
\end{threeparttable}
\end{table}

\begin{remark}
A policy maker interested in implementing a more equal allocation in equilibrium may decide to redistribute in two ways: either redistributing endowments from agents who have high network strength to those who have low network strength, or redistributing network strength (changing the network) from agents with high initial endowments to agents with low initial endowments.
\end{remark}

\section{Effect of higher order structures}

\noindent
{Higher order structures often capture important properties of the network and of the dynamical process \cite{Salnikov2019}. We already took into account the role of triangles elsewhere in this paper, here we focus on two types higher order structures on 4 nodes: tetrahedra  where only one node has endowment $e_1$  ($e_2$) while the remaining nodes have $e_2$ ($e_1$) and tetrahedra where two nodes have  endowment $e_1$  ($e_2$) and the other two have $e_2$ ($e_1$). For each we compute an index analogous to the clustering coefficient, which gives the fraction of higher order structures in the network over the possible number of high order structures of that dimension. Given that the two measures are correlated (Pearson's 0.59), we include each of them separately in the regression}

\begin{table}[htb]
    \centering
\sisetup{table-number-alignment = center, % <-- added/changed
         table-space-text-pre ={(},
         table-space-text-post={\textsuperscript{***}},
         input-open-uncertainty={[},
         input-close-uncertainty={]},
         table-align-text-pre=false,
         table-align-text-post=false}
\begin{threeparttable}
    \caption{Effect of higher order structures on post-trade utility}
    \label{tabhon1}
\begin{tabular}{r 
                S[table-format=-2.3] % <-- adopted to number of digits in numbers in cells
                S[table-format=-1.4] % <-- adopted ...
                S[table-format=-1.5] % <-- adopted ...
                 }
\toprule
                    & \multicolumn{3}{c}{Total utility gain}                                \\
\midrule
Intercept        &   -0.5019***   & (0.006)     \\

\midrule
Tetrahedra (1 different endowment)    &   0.0162\tnote{***}            &  (0.004)  \\
Assortativity   &        -0.0308***         &  (0.001)             \\
Connected   &        0.0941***         &  (0.002)             \\
Endowments similarity    &        0.3656***         &  (0.001)             \\
Number of nodes    &        0.1028***         &  (0.001)             \\
\midrule
Observations        & {61367}            &                           \\
R-squared        & 0.817           &         \\
Joint significance (p-value F-statistics) & 0.00 &    \\
\bottomrule
\end{tabular}
    \smallskip
    \footnotesize
standard error in parentheses\par
\begin{tablenotes}[para,flushleft]
    \item[*]    $p < 0.10$,
    \item[**]   $p < 0.05$,
    \item[***]  $p < 0.01$
    \end{tablenotes}\par
\end{threeparttable}
\end{table}

\begin{table}[htb]
    \centering
\sisetup{table-number-alignment = center, % <-- added/changed
         table-space-text-pre ={(},
         table-space-text-post={\textsuperscript{***}},
         input-open-uncertainty={[},
         input-close-uncertainty={]},
         table-align-text-pre=false,
         table-align-text-post=false}
\begin{threeparttable}
    \caption{Effect of higher order structures on post-trade utility}
    \label{tabhon2}
\begin{tabular}{r 
                S[table-format=-2.3] % <-- adopted to number of digits in numbers in cells
                S[table-format=-1.4] % <-- adopted ...
                S[table-format=-1.5] % <-- adopted ...
                 }
\toprule
                    & \multicolumn{3}{c}{Total utility gain}                                \\
\midrule
Intercept        &   -0.5017***   & (0.006)     \\

\midrule
Tetrahedra (2 different endowments)    &   0.0255\tnote{***}            &  (0.006)  \\
Assortativity   &        -0.0308***         &  (0.001)             \\
Connected   &        0.0941***         &  (0.002)             \\
Endowments similarity    &        0.3648***         &  (0.001)             \\
Number of nodes    &        0.1028***         &  (0.001)             \\
\midrule
Observations        & {61367}            &                           \\
R-squared        & 0.817           &         \\
Joint significance (p-value F-statistics) & 0.00 &    \\
\bottomrule
\end{tabular}
    \smallskip
    \footnotesize
standard error in parentheses\par
\begin{tablenotes}[para,flushleft]
    \item[*]    $p < 0.10$,
    \item[**]   $p < 0.05$,
    \item[***]  $p < 0.01$
    \end{tablenotes}\par
\end{threeparttable}
\end{table}

\noindent
{The results reported in tables \ref{tabhon1},\ref{tabhon2} show that both types of higher order structure are positively and significantly correlated with total utility gain, and that there is effectively no difference if we include one or the other in the regression. Similarly with the clustering coefficient for dissimilar triangles we can check that both higher order structures decrease inequality after trade, but even in this case density explains more of the variance than both higher order indices.}

\section{Details of the simulations} \label{append}

\noindent
{In section \ref{simulations} we consider as starting point to generate the weighted networks in our experiments the set of simple non isomorphic graphs on $n \in [3,7]$ nodes, and then we compute all permutations of endowments for each graph in this set. We did this operation for the sake of computational speed, but one shortcoming of this approach is that some of these permutations are redundant as they are equivalent in terms of initial configuration of the trade dynamics. To clarify this, recall that nodes in a graph can be grouped into orbits with respect to the graph automorphisms. Orbits identify the "role" of nodes in the graph: for example in a star, there are 2 roles, the hub and the periphery. Figure 5 shows two connected graphs on 4 nodes, one showing 3 node orbits the other 2. Node orbits are important because help us identify those networks which are equivalent in terms of trade dynamics. Take as example the leftmost graph in figure \ref{isom}, and consider the case in which 1 agent has endowment $(1,2)$ and 3 agents have endowment $(2,1)$. There are only three different trade configurations: one in which either of the two green nodes has endowment (1,2), one in which the red node has (1,2) and finally one in which the blue node has endowment (1,2), as can be seen in figure \ref{roles}. 
}

\begin{figure}[h!]
  \centering
  \begin{minipage}[b]{0.4\textwidth}
    \includegraphics[width=\textwidth]{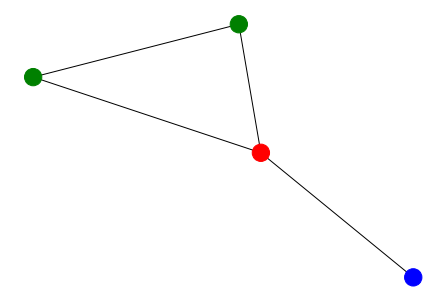}
    %\caption{Flower one.}
  \end{minipage}
  \begin{minipage}[b]{0.4\textwidth}
    \includegraphics[width=\textwidth]{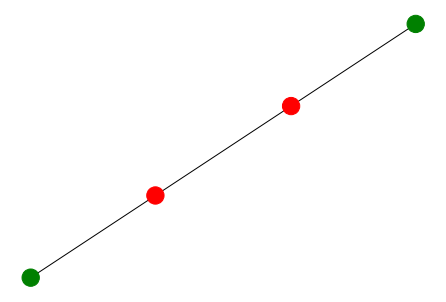}
   % \caption{Flower two.}
  \end{minipage}
   \caption{Two examples of connected non-isomorphic graphs on 4 nodes where nodes with the same orbit have the same colour}
   \label{isom}
\end{figure}

\begin{figure}[h!]
  \centering
  \begin{minipage}[b]{0.3\textwidth}
    \includegraphics[width=\textwidth]{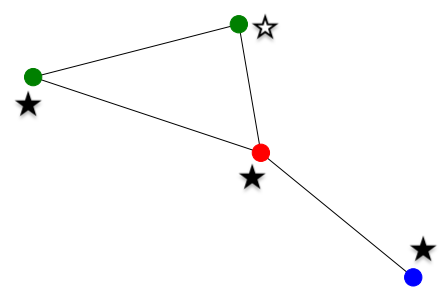}
    %\caption{Flower one.}
  \end{minipage}
  \begin{minipage}[b]{0.3\textwidth}
    \includegraphics[width=\textwidth]{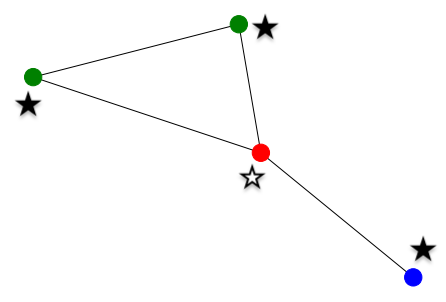}
  \end{minipage}  
  \begin{minipage}[b]{0.3\textwidth}
    \includegraphics[width=\textwidth]{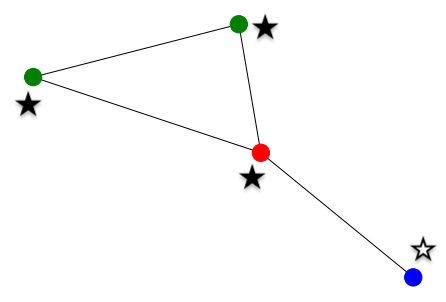}
  \end{minipage}
   \caption{The three different trade configurations for the graph with 3 different node roles when there is one different endowment. Endowment types are represented by filled or empty stars. Note that in the leftmost graph, permuting the endowment between the two green nodes makes no difference for the trade dynamics.}
   \label{roles}
\end{figure}

\newpage
\section{Further numerical examples} \label{append4}
\noindent
In this section we discuss some additional numerical examples on networks with 3 nodes.
We have three agents with Cobb-Douglas utility function  with constant return to scale: $U_i(x)=x_1^{\alpha_i}x_2^{1-\alpha_i}$. 
Call $\alpha_i$ the exponent of the utility function for agent $i$, and $x_i(0)=(x_{i,1}, x_{i,2})$ the initial allocation for agent $i$.
The network is represented by a unitary 2-simplex, where the barycentric coordinates of a point represent the weight of each edge.

Let us start from the example in Figure \ref{fig:simp}. For the same case, figure \ref{fig:ut_sim} shows the projections of the set of equilibria on the planes of the utility of two agents respectively, and makes the homeomorphism more evident.
\begin{figure}[h] 
\centering
\includegraphics*[width=8cm]{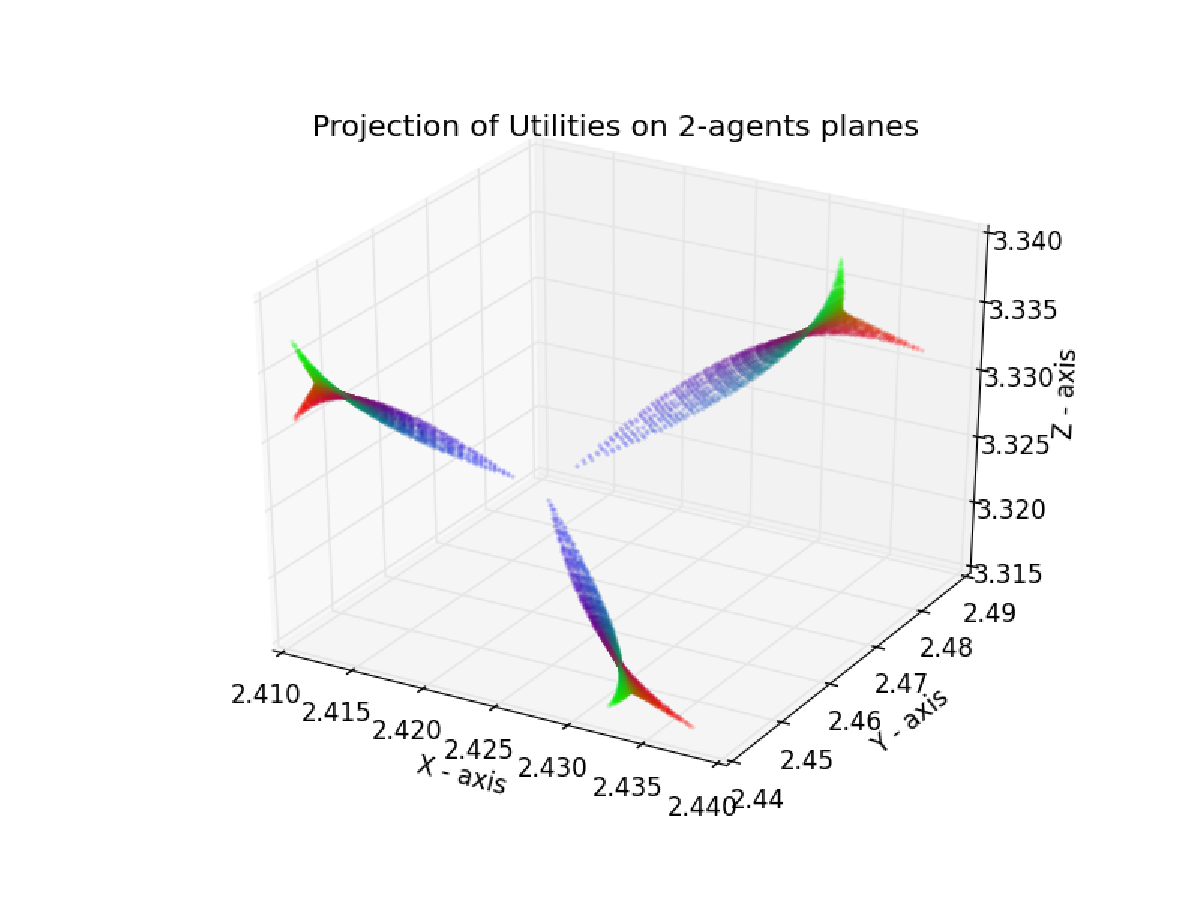}
\caption{Projection of equilibria in the space of utility on agents' planes }
\label{fig:ut_sim}
\end{figure}
Agent 3 has the highest initial endowment, and ends up having the highest level of utility in all the possible cases, ranging from 3.315 in its minimum, when the network is a star in which agent 2 is the core (blue vertex), to 3.330 in its maximum (when agent 3 is the core of the star). From this we can infer that the trade with agent 1 is the most advantageous for agent 3, as well as for agent 2, as also her utility hits the minimum point when she cannot trade with 3, and then increases when they trade on networks in which most of the interactions are between 2 and 3 (there is higher weight on this edge, as represented in the blue area). Clearly there is an asymptote in the growth of agent 1 utility moving towards a star in which agent 3 is the core (green area) and viceversa for agent 3 moving towards a star for which 1 is the core (red area). Looking at Figure \ref{fig:ut_sim}, utility of agent 1 is represented on the $x$ axis, and utility of agent 3 on $z$ axis: the figure has a twist in correspondence of the green area, where the utility of 1 stabilises around 2.430 and utility of 3 steeply increases till its maximum, while in correspondence of the red area utility of 1 stabilises around 3.330 while utility of 1 reaches its maximum.

\noindent
In Figure \ref{fig:eq_5_ut} we start from a different point in the space of goods, keeping the same utility functions. The initial allocations are such that $x_{1,1}>x_{2,1}>x_{3,1}$ and $x_{2,2}>x_{1,2}>x_{3,2}$ that is agent 1 and 2 have a lot of both goods and agent 3 is the poorest in both goods. As before each agent maximises her utility gain when she is the core of a star. Agent 3 is the one who is worse off by being a peripheral node when agent 1 or agent 2 are the core. This is not surprisingly as she is the one with the worst initial allocation. Viceversa utilities for agents 2 and 3 hit their minimum when agent 3 is in the core. By going towards the points in which the frequency of trades is mainly between agents 1 and 2 (the networks represented by the edge between the red and blue vertices in Figure \ref{fig:simplex}) their utility is close to the maximum, meaning that both rich agents would prefer trading among themselves because they can extract more utility, instead of trading with the \emph{poor} agent only.

\noindent
In Figure \ref{fig:eq_5_gsp} it is possible to observe the shape of the equilibrium points in the space of commodity one and commodity two respectively, holding the other commodity constant. As we would expect this is also a curved simplex, with each agent getting the highest quantity of each commodity (the vertices of the curved simplex) when they are the core of a star network.

\begin{figure}[h] 
\centering
\includegraphics*[width=8cm]{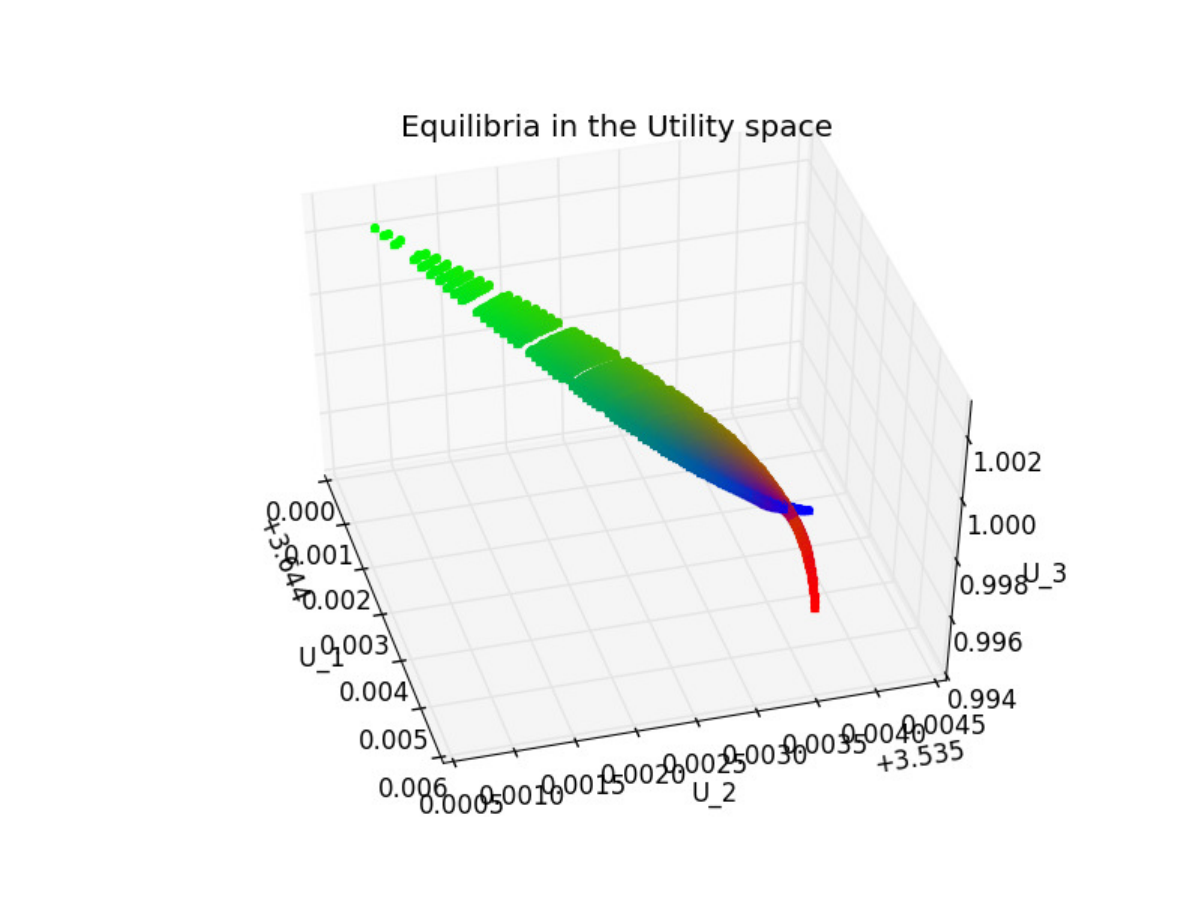}
\includegraphics*[width=8cm]{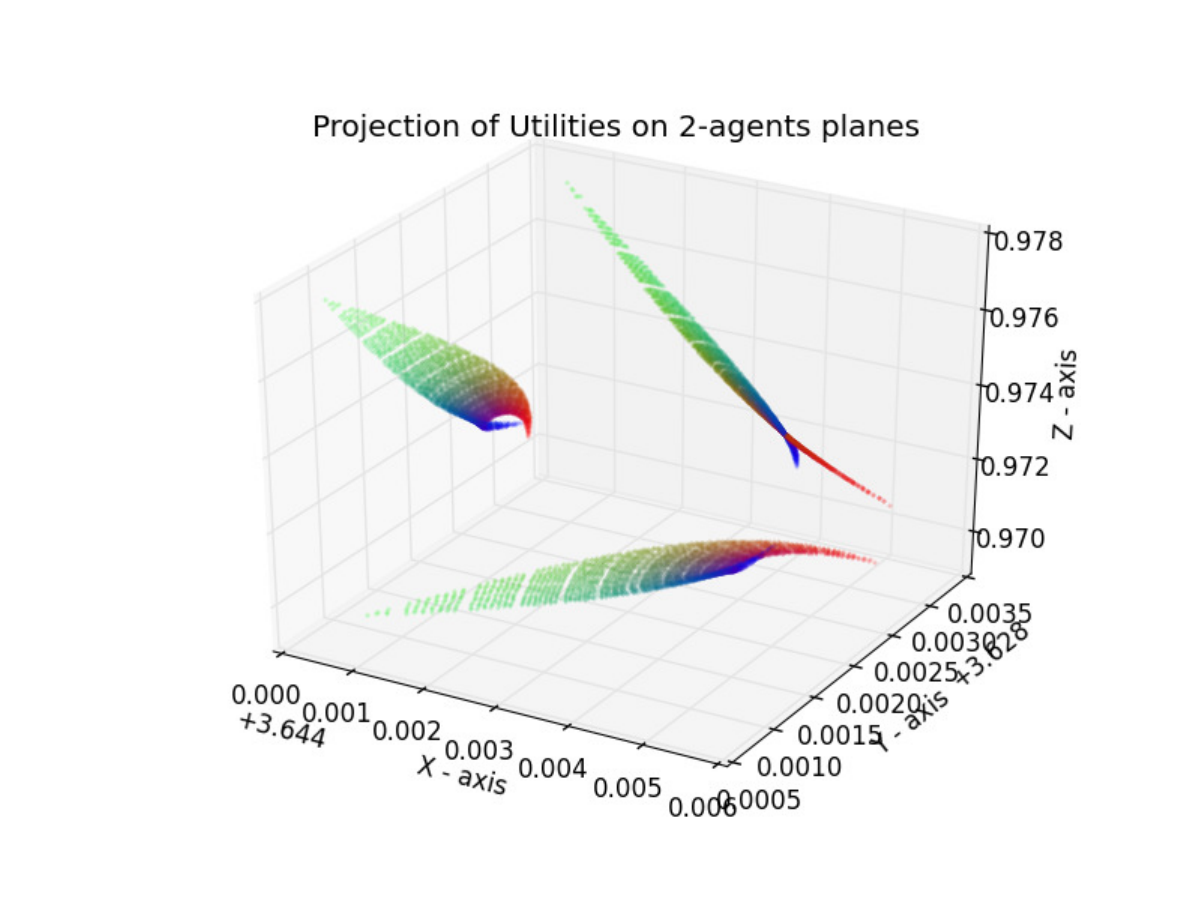}
\caption{Equilibria of the fair trading represented on the space of utilities for the case $\alpha_1=\alpha_2=\alpha_3=0.5$ (left) and projection on two-agents' planes.}
\label{fig:eq_5_ut}
\end{figure}

\begin{figure}[h] 
\centering
\includegraphics*[width=8cm]{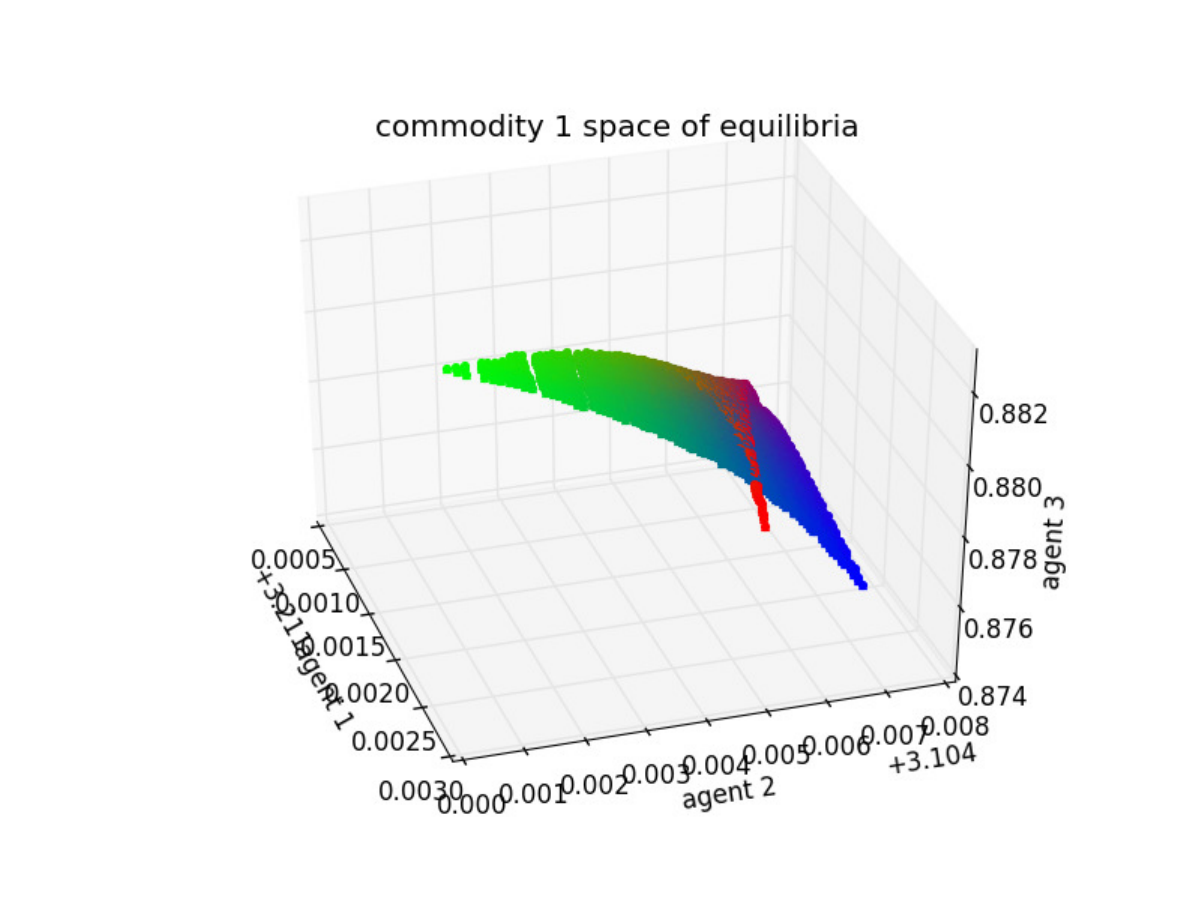}
\includegraphics*[width=8cm]{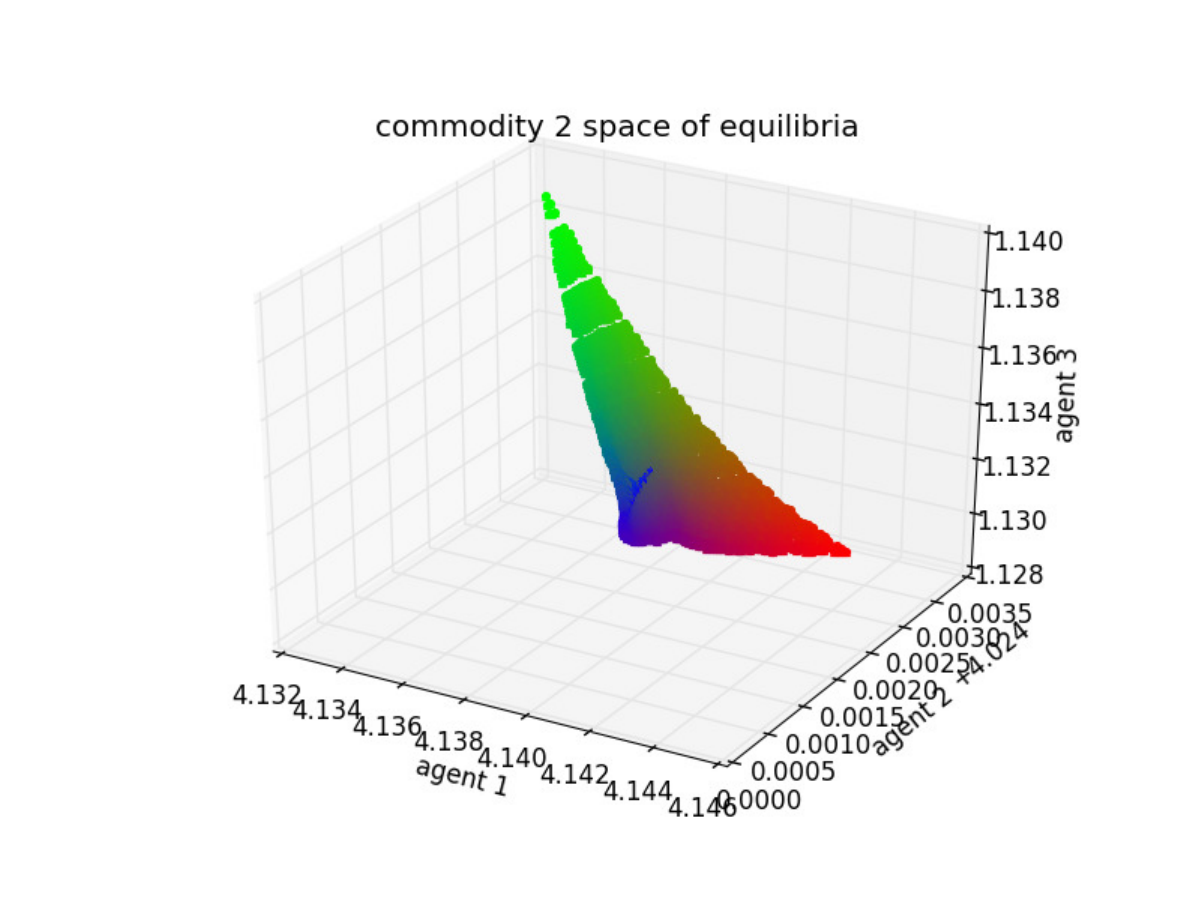}
\caption{Set of equilibria of a fair trading on the space of one commodity only}
\label{fig:eq_5_gsp}
\end{figure}

\begin{figure}[h] 
\centering
\includegraphics*[width=8cm]{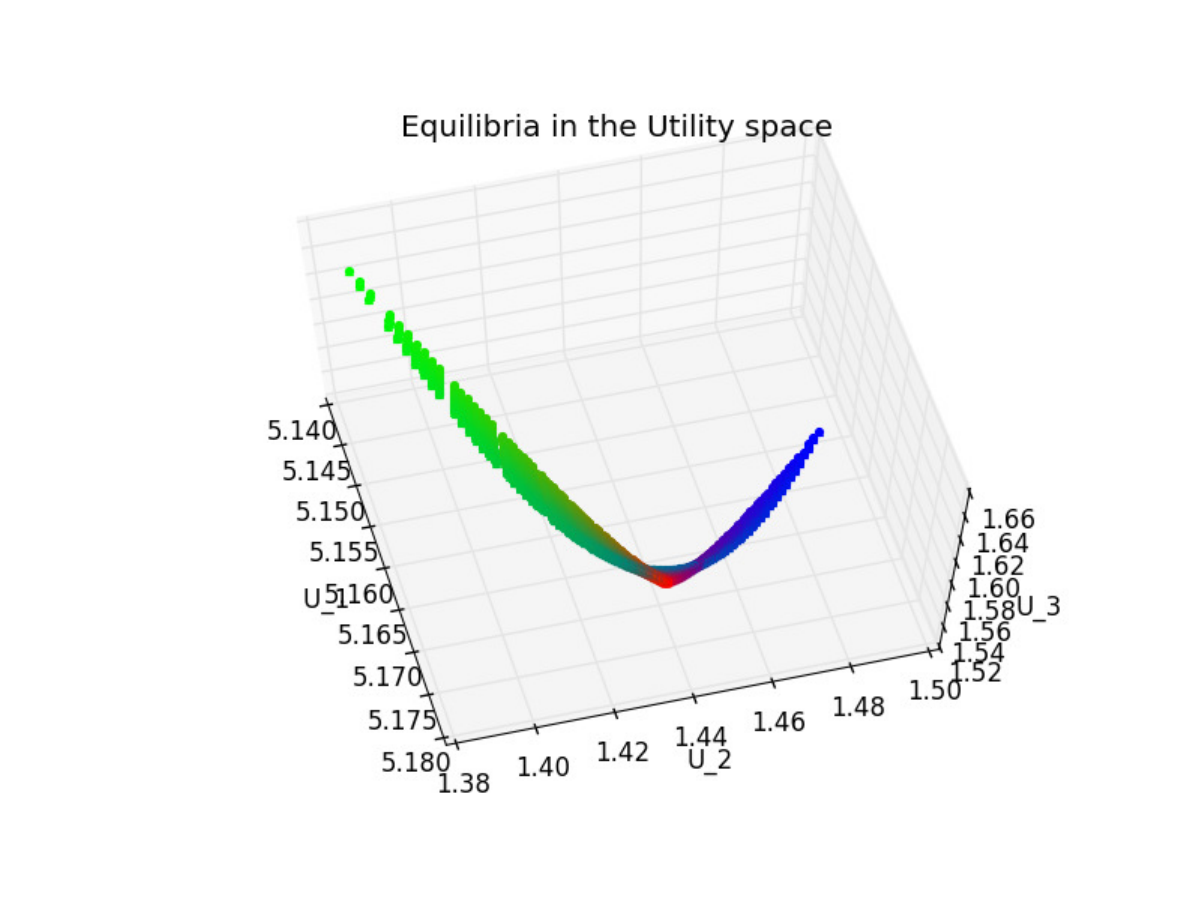}
\includegraphics*[width=8cm]{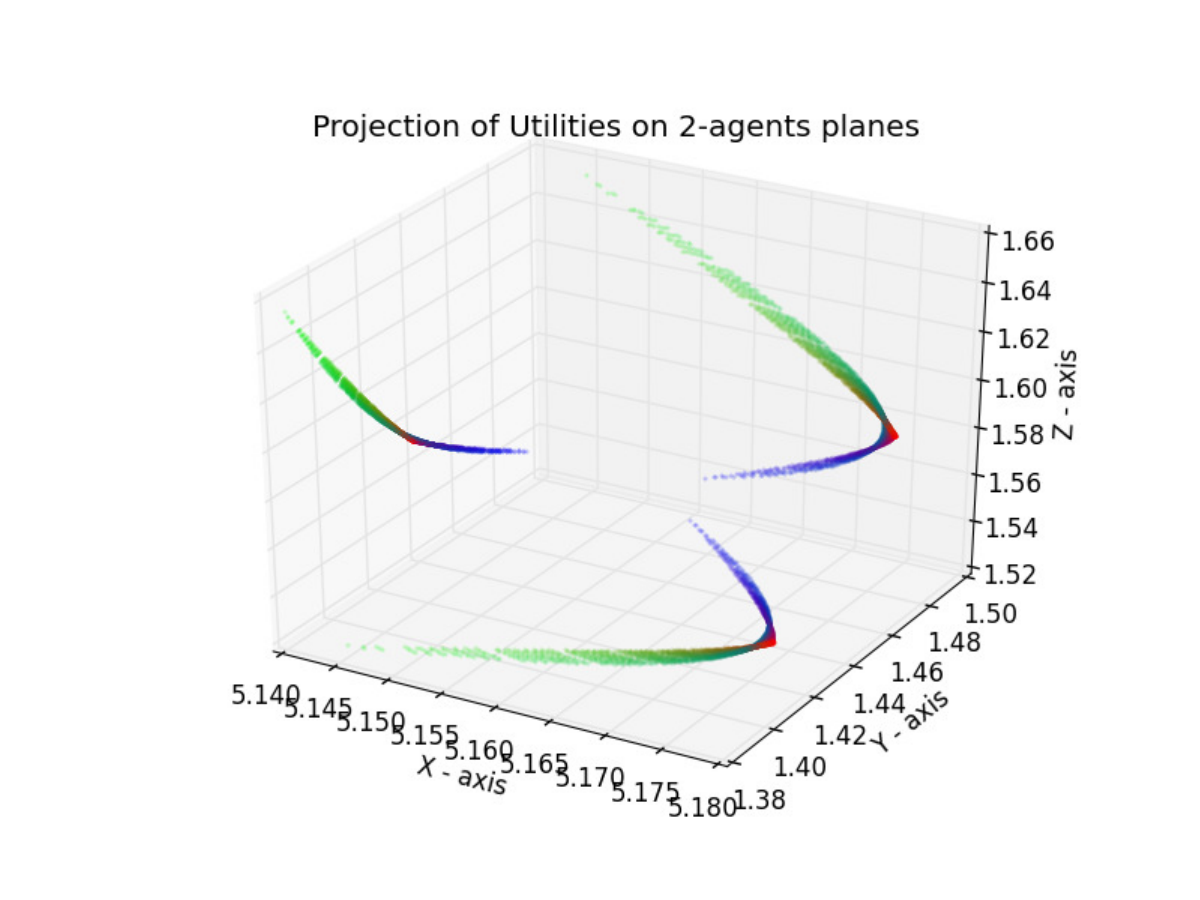}
\caption{ $\alpha_1=\alpha_2=\alpha_3=0.5$, extreme inequality: agent 1 is rich agents 2,3 are poor.}
\label{fig:eq_6_ut}
\end{figure}

\begin{figure}[h] 
\centering
\includegraphics*[width=8cm]{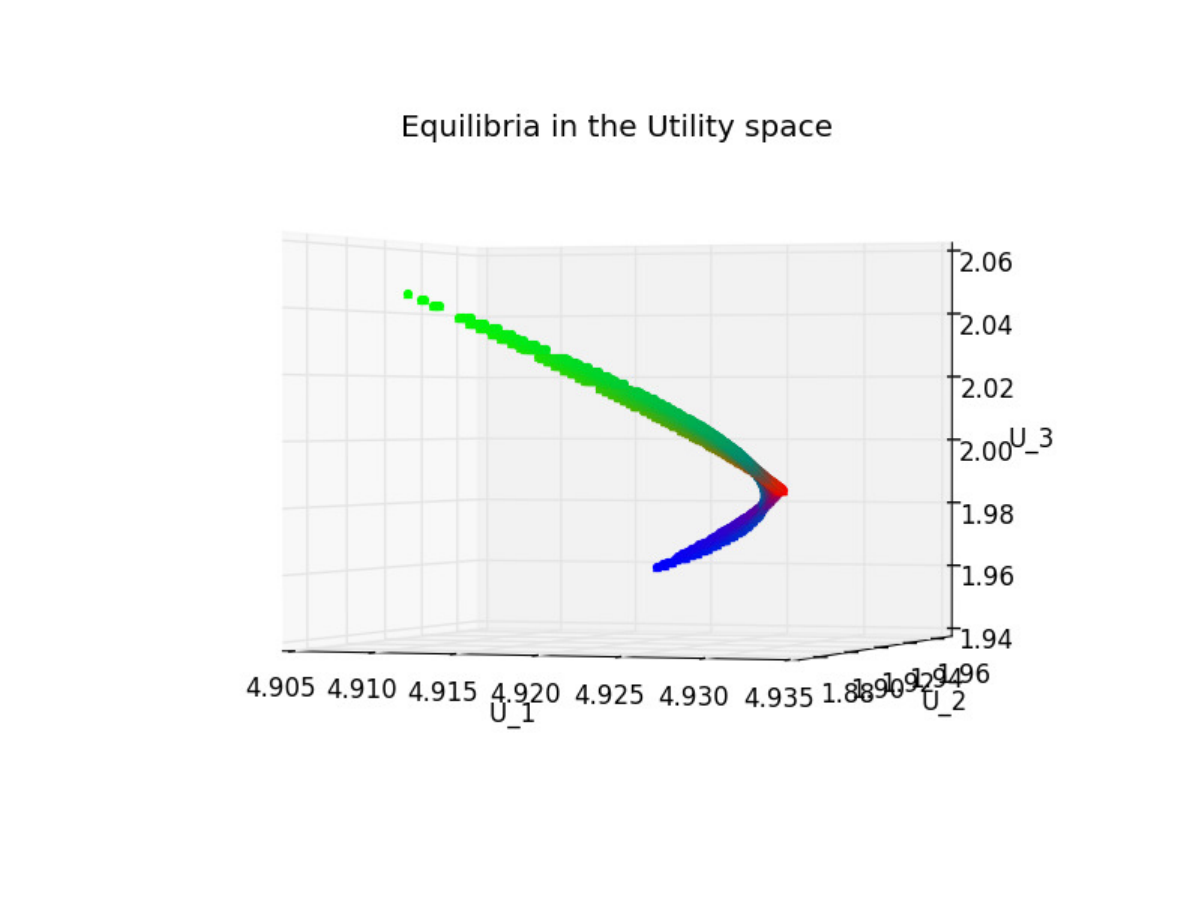}
\includegraphics*[width=8cm]{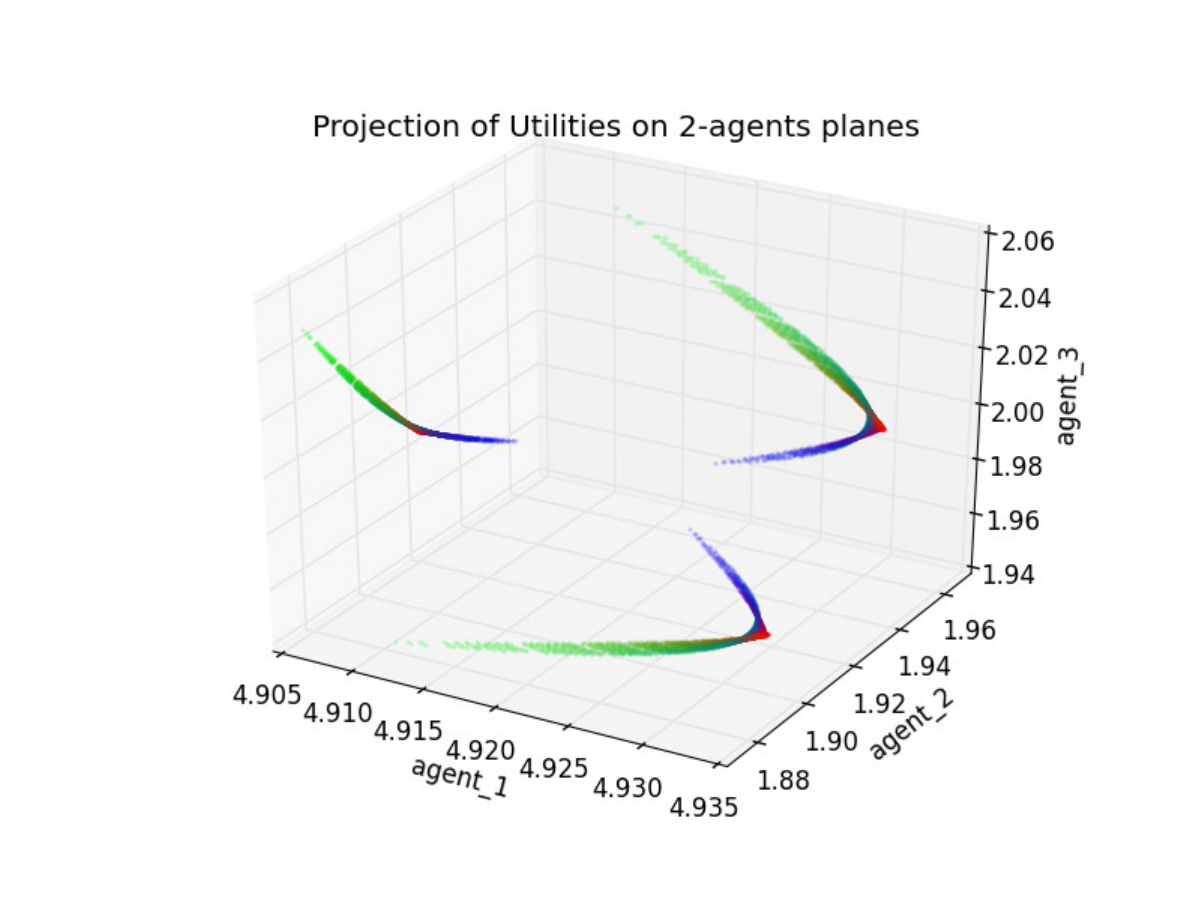}
\caption{$\alpha_1=\alpha_2=\alpha_3=0.2$, extreme inequality: agent 1 is rich agents 2,3 are poor}
\label{fig:eq_7_ut}
\end{figure}

\noindent
We then consider the case of extreme inequality in which agent 1 starts with a lot of both goods and agents 2 and 3 have a much inferior initial allocation, more precisely $x_{1,1}>x_{3,1}>x_{2,1}$ and $x_{1,2}>x_{2,2}=x_{3,2}$, results are represented in Figure \ref{fig:eq_6_ut} for the case of a Cobb-Douglas with $\alpha_1=\alpha_2=\alpha_3=0.5$, and in Figure \ref{fig:eq_7_ut} for the case in which they all prefer good 2 than good 1, that is their utility functions are such that $\alpha_i=0.2$ for $i={1,2,3}$. The first thing that we can notice is that the space of equilibria looks relatively similar in both cases, so that the initial allocation matters much more than the preferences, provided preferences are homogeneous across agents. Given the disproportion in initial allocations utility of agent 1 is greater than the two ``poor'' agents for all possible networks, while agent 2 and 3 maximize their utility when they are the core of a weighted star, as expected. Nonetheless note that both agent 2 and 3 will prefer to be in the periphery of the star where agent 1 is the core than being in the periphery of the star where any of the other "poor'' agent is in the core, even if the richest agent is maximizing her utility in this case. This is because both agents 2 and 3 prefer to have a consistent number of trades with agent 1, that is they will always prefer to trade in networks in which the weight of the edge connecting them with agent 1 is higher \emph{ceteris paribus}, and this determines the "boomerang'' shape of the set of equlibria.

\begin{figure}[h] 
\centering
\includegraphics*[width=8cm]{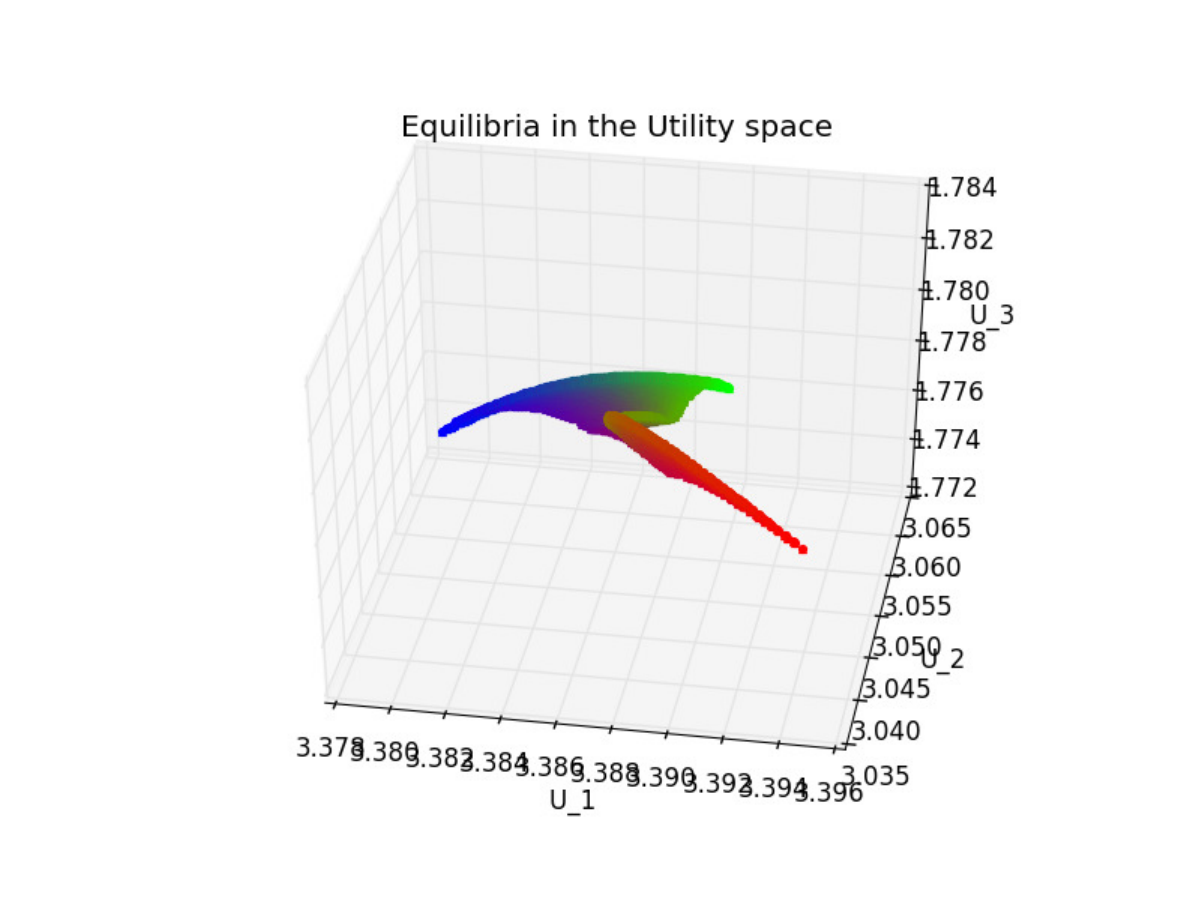}
\includegraphics*[width=8cm]{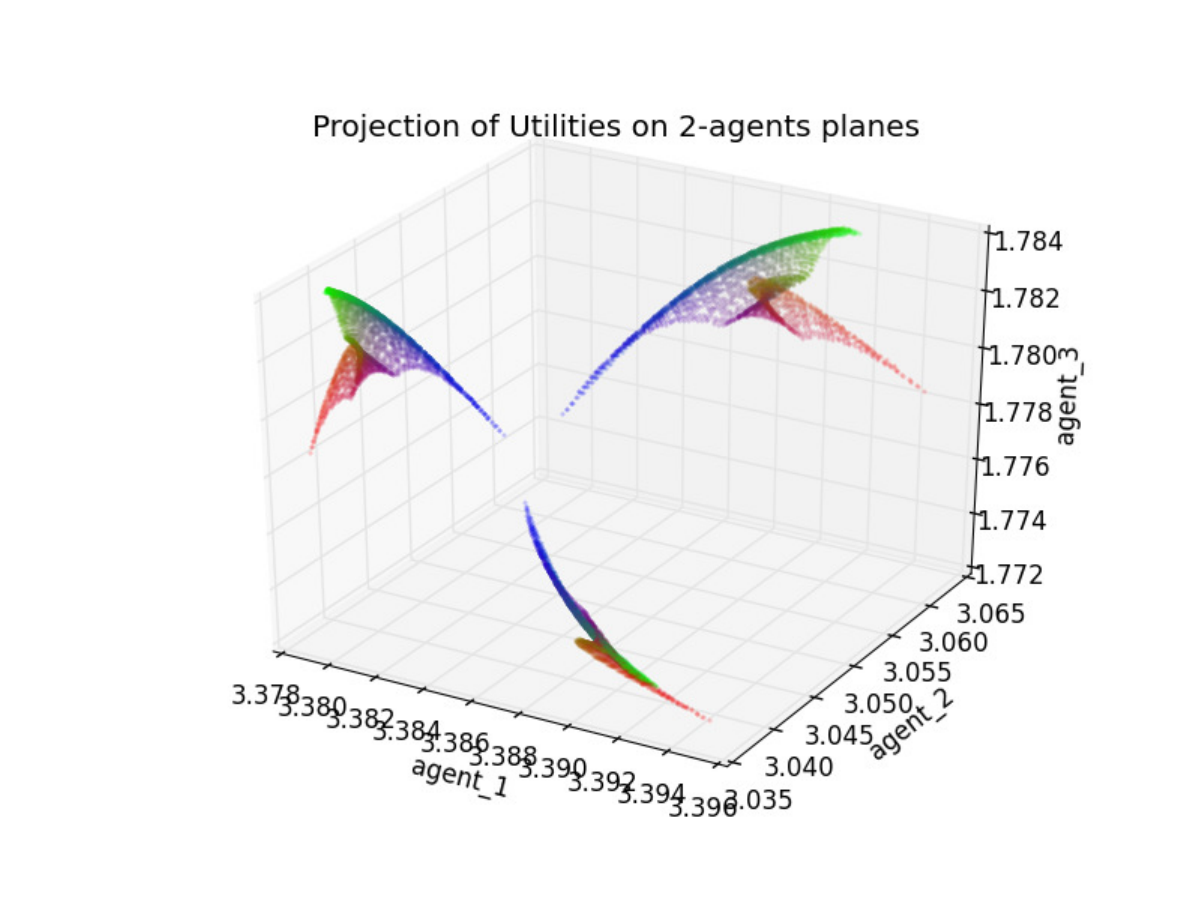}
\caption{$\alpha_1=\alpha_2=\alpha_3=0.2$ moderate inequality: agent 1 richer than agents 2 and 3.}
\label{fig:eq_8_ut}
\end{figure}

\noindent
Now consider the case in which agent 1 is still the richest, but the initial allocation is much less unequal than the previous two cases. The initial allocations in this case are $x_{1,1}>x_{2,1}>x_{3,1}$ and $x_{1,2}>x_{2,2}>x_{3,2}$, so agent 3 is the poorest. The results are represented in Figure \ref{fig:eq_8_ut}, preferences are the same as before. We can see how the picture drastically changes: now agent 2 worst position is when she is a peripheral node of a star where agent 1 is the core, and the higher the frequency of trade in which agent 1 is involved, the lower agent's 2 utility. Agent 3, the most disadvantaged, is worst off when she is in the periphery of a star with 2 in the core, she would rather prefer agent 1 to be the core. In general her utility will decrease the higher the weight on the edge between 2 and 1.

\end{document}